\documentclass[11pt]{article}
\usepackage[margin=1in]{geometry}
\usepackage{amsmath,amsthm,mathtools}
\usepackage{fontspec}
\usepackage{unicode-math}
\usepackage{microtype}
\usepackage{setspace}
\usepackage{graphicx}
\usepackage{float}
\usepackage{booktabs,longtable,array,calc,tabularx}
\usepackage{enumitem}
\usepackage{caption}
\usepackage[hidelinks]{hyperref}
\usepackage{fancyhdr}
\usepackage{xcolor}
\usepackage{xurl}
\usepackage{etoolbox}

\graphicspath{{figures/}}
\setlist[itemize]{leftmargin=1.6em,itemsep=2pt,topsep=4pt}
\setlist[enumerate]{leftmargin=1.8em,itemsep=2pt,topsep=4pt}

\newcommand{\E}{\mathbb E}
\newcommand{\Cov}{\operatorname{Cov}}
\newcommand{\TV}{\operatorname{TV}}
\newcommand{\logit}{\operatorname{logit}}
\newcommand{\expit}{\operatorname{expit}}

\newtheoremstyle{paperplain}
  {7pt}{7pt}{\itshape}{} {\bfseries}{.}{0.5em}{}
\theoremstyle{paperplain}
\newtheorem{proposition}{Proposition}
\newtheorem{corollary}{Corollary}[proposition]
\theoremstyle{remark}

\hypersetup{
  pdftitle={Beyond Actor-Partner Associations: A Target-Trial Framework for Causal Spillover Effects in Longitudinal Dyadic Data},
  pdfauthor={Subir Hait},
  pdfsubject={Causal inference for longitudinal dyadic data},
  pdfkeywords={actor-partner interdependence model, causal inference, dyadic data, interference, target trial, joint propensity score}
}
\begin{document}
\begin{titlepage}
\thispagestyle{empty}
\centering
\vspace*{1.25in}
{\LARGE\bfseries Beyond Actor-Partner Associations:\\[0.25em]
A Target-Trial Framework for Causal Spillover Effects\\[0.25em]
in Longitudinal Dyadic Data\par}
\vspace{1.0em}
{\large\itshape Joint Treatment Modeling, Doubly Robust Estimation, and an Application to Retirement Transitions in the Health and Retirement Study\par}
\vspace{2.4in}
{\large\bfseries Subir Hait\par}
\vspace{0.35em}
{Measurement and Quantitative Methods, Michigan State University\par}
\vfill
\end{titlepage}

\clearpage
\thispagestyle{fancy}
\begin{center}
{\large\bfseries Abstract}
\end{center}
\begin{quotation}\small
Actor-partner interdependence models (APIMs) describe how one member's exposure is associated with that member's and the partner's outcomes, but actor and partner coefficients are causal only under additional design and identification assumptions. We formulate APIM-style questions as dyadic target trials under within-dyad interference and make three links between the associational and causal formulations explicit. First, under joint exchangeability, positivity, and a homogeneous correctly specified treatment-related conditional mean, APIM actor, partner, interaction, and joint coefficients equal the corresponding causal contrasts. Second, for binary dyadic treatment, multiplying the true member-specific propensities induces an exact one-parameter discrepancy $\delta(H)=\Cov(A_1,A_2\mid H)$: the four joint-cell errors are $(+\delta,-\delta,-\delta,+\delta)$, total variation is $2|\delta|$, and IPW contrast bias has a closed-form representation. Third, the same discrepancy affects AIPW through its product with outcome-regression error. A 36-scenario simulation and a true-marginal factorization diagnostic show that the large separate-propensity bias under residual dependence is primarily factorization bias, while flexible AIPW reduces nonlinear-model bias but not weak-support instability. In the strongest dependence condition, factorizing even the true marginals reversed the spillover contrast and inflated the interaction to nearly five times its true value. We apply the framework to 2,400 Health and Retirement Study couples studying retirement transitions and later depressive symptoms. Joint retirement was rare, residual treatment dependence was substantial, and estimates were sensitive to support and propensity flooring. The contribution is a design-and-diagnostic bridge from APIM questions to causal dyadic analyses, not a new general interference or AIPW theory.

\end{quotation}
\clearpage
\section{Introduction}\label{introduction}

Close relationships create a basic challenge for causal inference: two
people who are often represented as separate observations can share
causes, influence one another\textquotesingle s exposures, and affect
one another\textquotesingle s outcomes. In studies of couples, spouses
share environments, resources, health shocks, caregiving
responsibilities, and life transitions. A partner\textquotesingle s
exposure can therefore be both a predictor of the other
partner\textquotesingle s outcome and a component of the intervention
itself. These features are substantively central to family, aging, and
health research, but they make the move from longitudinal association to
causal effect unusually delicate.

The actor-partner interdependence model (APIM) is a standard statistical
framework for dyadic data because it separates an
individual\textquotesingle s own-exposure association (the actor path)
from the partner-exposure association (the partner path) while
accounting for nonindependence within the dyad (Kenny, Kashy, \& Cook,
2006). Extensions include longitudinal and mediation formulations
(Ledermann, Macho, \& Kenny, 2011) and multivariate latent-variable
formulations (Kim \& Kim, 2024). APIM is exceptionally useful for
describing interdependence. Yet the labels actor effect and partner
effect are statistical labels: an APIM coefficient does not, by itself,
specify a counterfactual intervention or establish exchangeability of an
observational exposure.

The distinction is especially important when a partner coefficient is
interpreted as a spillover effect. A regression coefficient asks how a
partner\textquotesingle s measured exposure predicts a focal
member\textquotesingle s outcome after a specified adjustment. A causal
spillover contrast asks what would happen to the focal
member\textquotesingle s outcome if the partner\textquotesingle s
exposure were changed under a well-defined intervention while the focal
member\textquotesingle s own exposure were held fixed, or while both
exposures were governed by a joint policy. The two quantities coincide
only under assumptions that are stronger than those required to fit an
APIM.

Causal inference under interference supplies the appropriate
potential-outcome language. Under partial interference, treatments may
affect outcomes within a prespecified cluster but not across clusters
(Hudgens \& Halloran, 2008; Tchetgen Tchetgen \& VanderWeele, 2012). A
couple is the smallest nontrivial partial-interference cluster: the
treatment is a two-component vector and each member\textquotesingle s
potential outcome can depend on both components. Existing work has
developed inverse-probability and doubly robust estimators under partial
interference (Liu et al., 2019), dynamic regimes on dyadic networks
(Mussavi Rizi, Dubin, \& Wallace, 2024), and semiparametric estimators
under heterogeneous partial interference (Qu, Xiong, Liu, \& Imbens,
2026). These contributions establish that direct and spillover estimands
and efficient AIPW estimators are not new in themselves.

A separate challenge is longitudinal treatment-confounder feedback.
Prior depression, work status, physical health, income, or caregiving
can affect later retirement or behavior; earlier exposure can then alter
the same variables that predict subsequent exposure and outcome.
Conventional regression that conditions on a treatment-induced
confounder can block part of the causal pathway or induce other forms of
bias. Longitudinal g-methods were developed precisely for such settings
(Robins, Hernan, \& Brumback, 2000), but this logic is rarely made
explicit when longitudinal APIMs are interpreted causally.

The target-trial framework offers a design-first way to connect these
literatures. Rather than beginning with a regression equation, the
analyst specifies the hypothetical trial that would answer the causal
question: eligibility, treatment strategies, assignment, time zero,
follow-up, outcome, causal contrast, and analysis plan (Hernan \&
Robins, 2016; Hernan, Dahabreh, Dickerman, \& Swanson, 2025). The TARGET
reporting statement similarly emphasizes transparent separation of the
target trial, observational emulation, and analysis (Cashin et al.,
2025). For dyadic data, this design step forces a question that is easy
to miss in an APIM: is treatment assigned to each person separately, or
is the scientific intervention actually a joint dyadic state?

This paper develops a practical bridge from APIM-style scientific
questions to causal dyadic estimands. The central claim is deliberately
narrow. We make three links explicit so that the bridge is auditable: a
coefficient-equivalence result for the homogeneous APIM special case; an
exact decomposition of the error created by factorizing a binary dyadic
treatment law into its two true marginals; and the corresponding AIPW
error identity showing how treatment-law error interacts with
outcome-model error. These identities are elementary consequences of the
four-state treatment representation, but they organize the causal
argument and separate factorization error from generic nuisance
misspecification. We then use simulations to verify the implications, a
target-trial protocol to expose design gaps, and an HRS application to
show how residual treatment dependence and joint positivity enter an
applied analysis. The paper does not propose a new general theory of
interference or a new AIPW estimator.

\subsection{Statistical interdependence is not the same as causal
interference}\label{statistical-interdependence-is-not-the-same-as-causal-interference}

Two distinct forms of dependence are often present in dyadic data, and
treating them as interchangeable creates conceptual ambiguity.
Statistical interdependence means that outcomes or residuals from two
members of a pair are correlated. APIM, multilevel models, generalized
estimating equations, and structural equation models are designed to
represent such dependence so that regression coefficients and standard
errors are appropriate for paired observations. Causal interference is a
different statement: it says that one member\textquotesingle s exposure
could change the other member\textquotesingle s potential outcome. A
dataset can display strong residual correlation without meaningful
treatment interference, and it can display causal interference even
after residual correlation has been modeled correctly. The first problem
is primarily about the sampling and covariance structure; the second
changes the definition of the potential outcomes and therefore the
estimand itself.

This distinction matters because adding a random intercept, correlated
residual, or actor-partner covariance structure does not by itself
convert a partner regression coefficient into a spillover effect. Those
devices can model dependence conditional on the observed covariates, but
they do not specify what would happen under a joint intervention,
whether both members had positive probability of receiving each relevant
treatment state, or whether common causes of the two treatment
assignments and two outcomes were adequately controlled. Conversely,
potential-outcome notation alone does not solve the statistical
dependence problem: inference still needs to treat the dyad, rather than
the individual, as the independent sampling unit. A coherent causal
dyadic analysis therefore needs both pieces at once: a counterfactual
treatment definition and an analysis that respects within-dyad
dependence.

The four-state joint treatment representation is a particularly useful
bridge. It makes clear that an actor contrast, a partner contrast, and a
joint contrast are not three unrelated coefficients but contrasts among
the same set of regime-specific potential-outcome means. This
perspective also reveals why apparently adequate marginal treatment
frequencies can conceal poor causal support. If 20\% of individuals
retire during an interval, for example, it does not follow that a
sufficient number of couples experience the 11 state in which both
members retire. The support problem is joint rather than marginal.

\subsection{Three recurring routes from APIM language to causal
overinterpretation}\label{three-recurring-routes-from-apim-language-to-causal-overinterpretation}

The first route is terminological. Because APIM uses the words actor
effect and partner effect, substantive reports can slide from
\textquotesingle association conditional on the specified
model\textquotesingle{} to \textquotesingle causal
effect\textquotesingle{} without an explicit change in assumptions. The
same linguistic problem appears elsewhere in regression analysis, but it
is especially salient for partner effects because the causal analogue is
naturally described as a spillover or interference effect. The causal
question is not whether a partner variable has a nonzero coefficient; it
is whether the focal member\textquotesingle s outcome would differ under
two well-defined partner exposure assignments while the relevant own
exposure component is held fixed.

The second route is treatment-model factorization. When the causal
treatment is dyadic, the relevant generalized propensity score has four
categories. Estimating two individual binary propensities and
multiplying them is attractive because it reuses familiar methods, but
it embeds a substantive conditional-independence restriction. Couples
may coordinate retirement, treatment uptake, caregiving, diet, or health
behavior after conditioning on measured baseline history. In those
settings the product of two individually correct marginal models is not
necessarily the correct probability of the observed joint state. The
simulation study stress-tests this practical strategy as residual
treatment dependence increases; because the fitted separate-propensity
arm also uses parametric marginal working models, its bias is not
interpreted as a pure numerical decomposition of factorization error.

The third route is longitudinal overadjustment. Researchers studying
couples over repeated waves often adjust for the most recent
measurements of depression, health, caregiving, employment, or
relationship quality. Such adjustment is useful for prediction and may
be appropriate for certain controlled direct effects, but it can be
incompatible with the total effect of a sustained exposure regime when
the adjusted covariate is affected by earlier exposure. Robins et al.
(2000) showed why treatment-confounder feedback requires g-methods in
longitudinal causal inference. In a dyad the same problem occurs with a
larger history: each member\textquotesingle s treatment and intermediate
variables can enter the other\textquotesingle s subsequent assignment
and outcome process.

\subsection{Scope of the
contribution}\label{scope-of-the-contribution}

The framework developed here is intentionally modular. It does not
replace APIM and does not propose a new general theory of interference.
Its contribution is a translation-and-diagnostic methodology for
researchers whose scientific questions are phrased in actor-partner
language but whose intended conclusions are causal. The framework
identifies the regime means corresponding to those questions, gives an
exact special-case bridge from APIM coefficients to causal contrasts,
makes the consequences of treatment-law factorization algebraically
explicit, and embeds estimation in a target-trial protocol. The
simulations are therefore not a method tournament: they verify the
formal identities under controlled conditions, show what flexible
nuisance learning can and cannot repair, and demonstrate how weak joint
support limits all estimators.

\begin{figure}[H]
\centering
\includegraphics[width=0.88\textwidth]{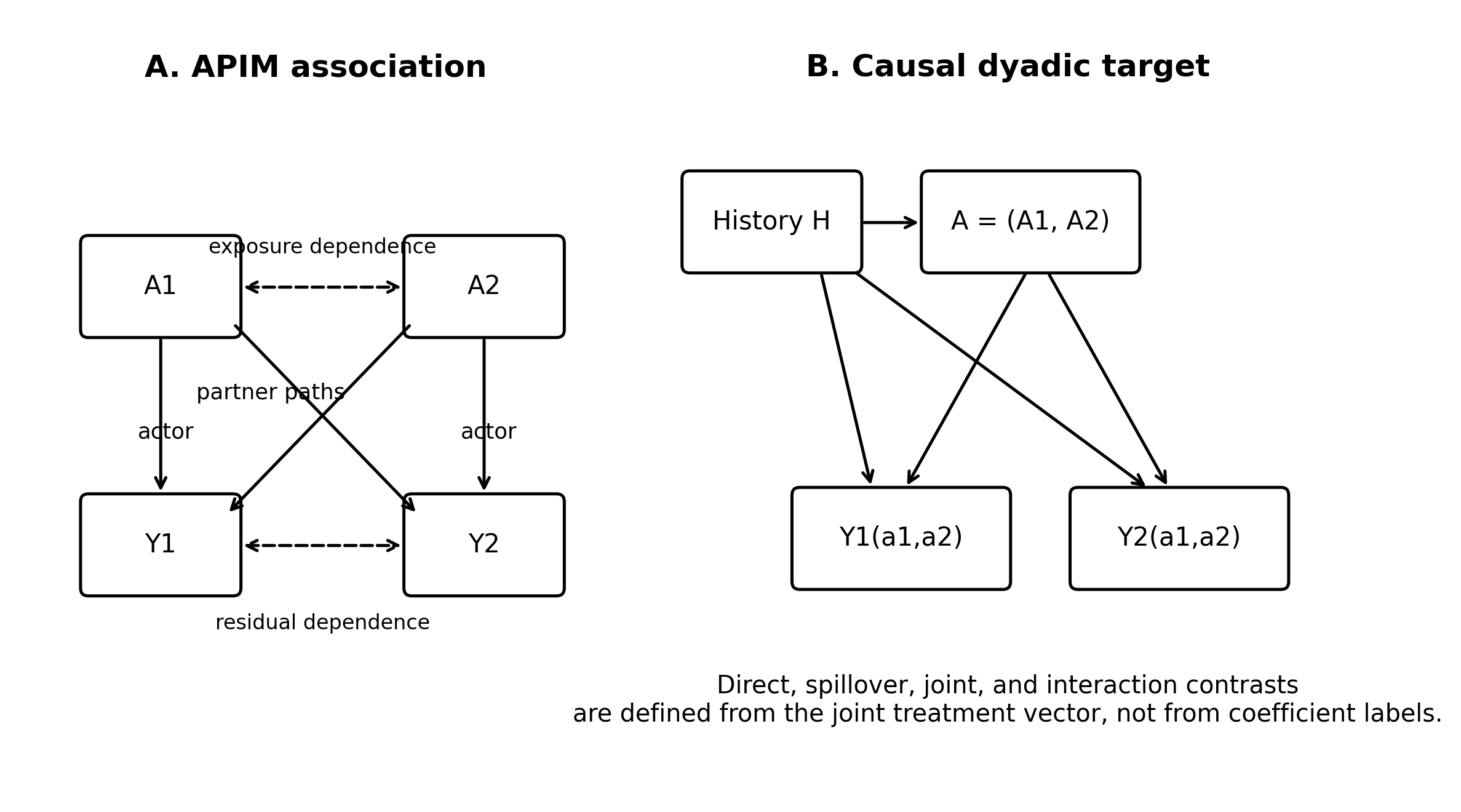}
\caption{From APIM association to a causal dyadic treatment target. APIM paths describe conditional associations. The causal formulation instead indexes each member-specific potential outcome by the joint treatment vector and defines contrasts across joint treatment states.}
\end{figure}

\section{Dyadic Causal Estimands and Their Relation to
APIM}\label{dyadic-causal-estimands-and-their-relation-to-apim}

\subsection{APIM as an associational
model}\label{apim-as-an-associational-model}

For dyad $d$ with members $j=1,2$, let $A_1$ and $A_2$ denote member-specific
binary exposures and $Y_1$ and $Y_2$ the corresponding outcomes measured after
exposure classification. A simple adjusted APIM-style conditional mean
can be written as
\begin{equation}
\begin{aligned}
\E(Y_1\mid A_1,A_2,H)&=\alpha_1(H)+\beta_A A_1+\beta_P A_2+\beta_I A_1A_2,\\
\E(Y_2\mid A_1,A_2,H)&=\alpha_2(H)+\beta_P A_1+\beta_A A_2+\beta_I A_1A_2,
\end{aligned}
\tag{1}
\end{equation}
where $H$ denotes measured histories of both members and shared dyadic
covariates. More general APIMs may allow distinguishable members,
member-specific coefficients, random effects, correlated residuals,
latent variables, or repeated measures. Those additions model dependence
and heterogeneity, but they do not by themselves turn $\beta_A$ or $\beta_P$ into
intervention effects.

\subsection{Potential outcomes under within-dyad
interference}\label{potential-outcomes-under-within-dyad-interference}

Let $A=(A_1,A_2)$ take one of four states $\mathcal S=\{00,10,01,11\}$. For member $j$,
define $Y_j(a_1,a_2)$ as the potential outcome that would be observed under
joint treatment state $(a_1,a_2)$. No between-dyad interference means that
$Y_j$ depends on the treatment of the two members of its dyad but not on
treatment vectors assigned to other dyads. This is a two-person special
case of partial interference.

For member 1, four elementary contrasts are:

\[
D_1=\E\{Y_1(1,0)-Y_1(0,0)\},
\]

\[
S_1=\E\{Y_1(0,1)-Y_1(0,0)\},
\]

\[
J_1=\E\{Y_1(1,1)-Y_1(0,0)\},
\]

\begin{equation}
I_1=\E\{Y_1(1,1)-Y_1(1,0)-Y_1(0,1)+Y_1(0,0)\}.
\tag{2}
\end{equation}

Analogous contrasts are defined for member 2. In indistinguishable or
symmetrically reported dyads, actor and partner effects can be averaged
across both orientations. For example, the symmetric actor effect is one
half of the member-1 direct contrast and the corresponding member-2
contrast after switching the role labels. This prevents arbitrary member
ordering from becoming a scientific estimand.

\textbf{Table 1. Causal dyadic estimands and APIM analogues}

\begin{longtable}[]{@{}
  >{\raggedright\arraybackslash}p{(\columnwidth - 6\tabcolsep) * \real{0.2500}}
  >{\raggedright\arraybackslash}p{(\columnwidth - 6\tabcolsep) * \real{0.2500}}
  >{\raggedright\arraybackslash}p{(\columnwidth - 6\tabcolsep) * \real{0.2500}}
  >{\raggedright\arraybackslash}p{(\columnwidth - 6\tabcolsep) * \real{0.2500}}@{}}
\toprule\noalign{}
\begin{minipage}[b]{\linewidth}\raggedright
\textbf{Estimand}
\end{minipage} & \begin{minipage}[b]{\linewidth}\raggedright
\textbf{Potential-outcome definition (member 1)}
\end{minipage} & \begin{minipage}[b]{\linewidth}\raggedright
\textbf{Causal interpretation}
\end{minipage} & \begin{minipage}[b]{\linewidth}\raggedright
\textbf{APIM analogue}
\end{minipage} \\
\midrule\noalign{}
\endhead
\bottomrule\noalign{}
\endlastfoot
Actor/direct & E{[}Y1(1,0)-Y1(0,0){]} & Own exposure changes; partner
fixed at 0 & Actor coefficient under restrictive conditions \\
Partner/spillover & E{[}Y1(0,1)-Y1(0,0){]} & Partner exposure changes;
own exposure fixed at 0 & Partner coefficient under restrictive
conditions \\
Joint & E{[}Y1(1,1)-Y1(0,0){]} & Move dyad from neither exposed to both
exposed & Sum of actor, partner, and interaction in an additive
conditional-mean model \\
Interaction & E{[}Y1(1,1)-Y1(1,0)-Y1(0,1)+Y1(0,0){]} & Departure from
additivity of the two exposure components & Actor-by-partner product
term under restrictive conditions \\
\end{longtable}

\emph{Note.} Member-2 contrasts are obtained by switching member labels.
The empirical HRS analysis reports symmetric averages over both member
orientations.

\subsection{When can an APIM coefficient be
causal?}\label{when-can-an-apim-coefficient-be-causal}

A useful clarification is that APIM coefficients are not inherently
noncausal. A causal interpretation is available in a restrictive but
transparent special case.

\begin{proposition}[APIM coefficient equivalence under homogeneous treatment effects]\label{prop:apim-equivalence}
Suppose consistency holds, there is no interference between dyads, joint conditional exchangeability holds given $H$, the required joint treatment states satisfy positivity, and the member-1 conditional mean satisfies
\[
m_1(a_1,a_2,H)=\alpha_1(H)+\beta_Aa_1+\beta_Pa_2+\beta_Ia_1a_2,
\]
where $\beta_A$, $\beta_P$, and $\beta_I$ do not vary with $H$. Then the standardized member-1 direct, spillover, interaction, and $11$-versus-$00$ joint causal contrasts equal the corresponding treatment coefficients and their sum:
\[
D_1=\beta_A,\qquad S_1=\beta_P,\qquad I_1=\beta_I,\qquad J_1=\beta_A+\beta_P+\beta_I.
\]
\end{proposition}

\begin{corollary}[Standardization under effect modification]\label{cor:standardization}
If treatment effects vary with $H$, a single APIM coefficient need not equal the marginal causal contrast. The causal quantity is obtained by evaluating $m_j(a,H)$ under each required joint state and averaging over the target $H$ distribution before taking the contrast. Thus $\alpha_j(H)$ may be arbitrarily complex in Proposition~\ref{prop:apim-equivalence}; the critical restriction is the homogeneous, correctly specified treatment-related part of the conditional mean. Omitting the $A_1A_2$ interaction additionally requires additivity if actor and partner coefficients are to reconstruct the joint contrast.
\end{corollary}

Proofs are given in Appendix~\ref{sec:proofs}.

\subsection{Distinguishable, indistinguishable, and symmetrically
reported
dyads}\label{distinguishable-indistinguishable-and-symmetrically-reported-dyads}

Dyadic analyses differ in whether member labels carry substantive
meaning. In distinguishable dyads, member roles are defined
independently of the observed exposure and outcome values. Examples
include patient-caregiver dyads, parent-child pairs, or couples in which
a role such as birth parent versus partner is scientifically relevant.
In that setting the member-1 direct effect and member-2 direct effect
are separate causal quantities and need not be equal. Likewise, the
effect of member 2\textquotesingle s exposure on member 1 can differ
from the reverse spillover effect. A causal analysis should preserve
that asymmetry rather than imposing an actor-partner equality for
convenience.

In indistinguishable dyads, or when the scientific report does not
privilege either orientation, a symmetric estimand is often preferable.
Let D1 denote member 1\textquotesingle s direct contrast and D2 the
corresponding contrast after swapping labels. The symmetric actor effect
is (D1+D2)/2, with analogous averages for spillover, joint, and
interaction contrasts. This final average is invariant to which spouse
is called member 1 even when the administrative labeling rule is
correlated with age, gender, or other characteristics. The
orientation-specific components, however, remain label-specific and
should not be interpreted as estimates for demographically exchangeable
roles unless that is substantively justified.

Symmetric reporting is not the same as assuming identical structural
equations for the two members. One can fit distinct member-specific
outcome regressions, generate two sets of regime means, and average only
at the final estimand stage. This separation is useful in
machine-learning implementations because it allows the outcome surface
for one member to differ from the other\textquotesingle s while
preserving an orientation-invariant report. If substantive
distinguishability is important, the averaging step should simply be
omitted and the member-specific contrasts reported separately.

\subsection{Conditional coefficients, marginal causal effects, and
effect
heterogeneity}\label{conditional-coefficients-marginal-causal-effects-and-effect-heterogeneity}

A second distinction concerns conditional versus marginal targets. APIM
coefficients are usually conditional on the variables included in the
model. The potential-outcome contrasts in Equation (2) are marginal
averages over a target population. These coincide under the homogeneous
linear conditional-mean structure stated in Section 2.3, but they need
not coincide when treatment effects vary with age, baseline health, sex,
prior outcome, or shared household history. In the presence of effect
modification, a coefficient from one regression parameterization is not
generally the same object as a standardized population-average causal
contrast.

The g-formula representation in Equation (3) resolves the distinction
explicitly. For every joint state a, the analyst first estimates the
conditional mean $m_j(a,H)$ and then averages predictions over the
empirical distribution of H in the target population. The contrast is
therefore standardized to a named population. This operation is routine
in modern causal inference but is easy to overlook in dyadic
applications, where interest often centers on path coefficients. The
target population should be stated because different eligibility rules,
couple types, or baseline waves can produce different marginal effects
even when the conditional response surface is unchanged.

This point also clarifies what cross-fitting contributes. Cross-fitting
is not itself a causal assumption and does not define the estimand. It
is an estimation device that permits more flexible nuisance models while
reducing overfitting bias in orthogonal estimating equations. The causal
target continues to be determined by the joint treatment regimes, target
population, time ordering, and identification assumptions. Keeping those
layers separate prevents machine-learning performance from being
mistaken for identification.

\section{A Dyadic Target-Trial
Formulation}\label{a-dyadic-target-trial-formulation}

\subsection{Protocol components}\label{protocol-components}

The target-trial formulation treats the couple as the unit of treatment
assignment while retaining member-specific outcomes. This distinction
does not require that both members receive the same treatment. It
requires that the intervention be represented by the full joint state
$A=(A_1,A_2)$, because each component can affect either outcome. Table 2
shows a generic protocol and the specification used in the HRS
illustration.

\textbf{Table 2. Dyadic target-trial protocol and HRS emulation}

\begin{longtable}[]{@{}
  >{\raggedright\arraybackslash}p{(\columnwidth - 4\tabcolsep) * \real{0.3333}}
  >{\raggedright\arraybackslash}p{(\columnwidth - 4\tabcolsep) * \real{0.3333}}
  >{\raggedright\arraybackslash}p{(\columnwidth - 4\tabcolsep) * \real{0.3333}}@{}}
\toprule\noalign{}
\begin{minipage}[b]{\linewidth}\raggedright
\textbf{Component}
\end{minipage} & \begin{minipage}[b]{\linewidth}\raggedright
\textbf{Hypothetical dyadic target trial}
\end{minipage} & \begin{minipage}[b]{\linewidth}\raggedright
\textbf{HRS emulation}
\end{minipage} \\
\midrule\noalign{}
\endhead
\bottomrule\noalign{}
\endlastfoot
Eligibility & Both members satisfy prespecified criteria using
pre-exposure information only. & At first eligible baseline from wave 4
onward, both members were age 50-69, working full- or part-time, linked
as a reciprocal couple, and had required baseline covariates and later
CES-D observed. \\
Treatment strategies & Assign joint state 00, 10, 01, or 11 over a
prespecified exposure interval. & At t+1, each member was classified as
still working full/part-time (0) or fully retired (1); other labor-force
states were excluded from this contrast. \\
Time zero & Eligibility and assignment are aligned. & Baseline t defines
eligibility/history; retirement state is classified at the next biennial
wave. \\
Outcome & Member-specific outcome under each joint treatment state. &
RAND CES-D for each member at t+2. \\
Causal contrasts & Direct, spillover, joint, and interaction contrasts.
& Symmetric averages over the two member orientations. \\
Analysis & Estimator compatible with joint assignment and confounding
structure. & Joint multinomial IPW, parametric AIPW, and 5-fold
cross-fitted random-forest AIPW; separate-PS IPW and APIM-style outcome
regression are comparators. \\
\end{longtable}

\subsection{Emulation-gap audit}\label{emulation-gap-audit}

A target-trial table is most useful when it records not only the
intended protocol but also where the observational emulation departs
from it. Table 2A makes those departures explicit for the HRS
illustration. The audit is descriptive rather than a bias-correction
procedure: its purpose is to prevent post-baseline selection,
exposure-window confounding, treatment-version ambiguity, and weak
support from being hidden inside an otherwise well-aligned regression
analysis. This design-first emphasis is complementary to the DYADIC
meta-framework, which classifies how an intervention is dyadic in who
participates, what is done, the mechanisms targeted, and the outcomes
considered (Berli et al., 2026); the present framework asks what joint
treatment vector and counterfactual contrast those intervention features
imply.

Table 2A. HRS target-trial emulation-gap audit

\begin{longtable}[]{@{}
  >{\raggedright\arraybackslash}p{(\columnwidth - 4\tabcolsep) * \real{0.3333}}
  >{\raggedright\arraybackslash}p{(\columnwidth - 4\tabcolsep) * \real{0.3333}}
  >{\raggedright\arraybackslash}p{(\columnwidth - 4\tabcolsep) * \real{0.3333}}@{}}
\toprule\noalign{}
\begin{minipage}[b]{\linewidth}\raggedright
\textbf{Design element}
\end{minipage} & \begin{minipage}[b]{\linewidth}\raggedright
\textbf{HRS operationalization}
\end{minipage} & \begin{minipage}[b]{\linewidth}\raggedright
\textbf{Emulation gap / consequence}
\end{minipage} \\
\midrule\noalign{}
\endhead
\bottomrule\noalign{}
\endlastfoot
\textbf{Eligibility / selection} & Baseline working couple, but the
binary contrast additionally requires an allowed labor-force state at
t+1 and observed CES-D at t+2. & Conditions on post-baseline states and
outcome observation; interpretation is for the selected analytic cohort
unless selection is otherwise addressed. \\
\textbf{Exposure interval} & Retirement is classified at the next
biennial wave rather than assigned at baseline. & Health or employment
shocks during the interval can affect both retirement and later CES-D;
baseline-only adjustment cannot remove such exposure-window
confounding. \\
\textbf{Treatment versions} & ``Fully retired'' is treated as a common
treatment label. & Retirement can arise through heterogeneous health,
financial, caregiving, or employer pathways; consistency requires
sufficient version comparability for the claimed contrast. \\
\textbf{Censoring / competing events} & Complete observed t+2 outcome is
required; death and attrition are not weighted. & Informative
nonobservation can change the target population and induce selection. \\
\textbf{Joint support} & Four retirement states are required; only 40
couples are observed in 11. & Finite-sample joint positivity is weak for
11-related joint and interaction contrasts, regardless of adequate
marginal retirement prevalence. \\
\end{longtable}

Note. The audit distinguishes the hypothetical target trial from the
operational choices used in the HRS illustration. These gaps do not by
themselves determine the direction or magnitude of bias; they identify
where additional assumptions or sensitivity analyses are needed.

\subsection{Identification}\label{identification}

Let H contain the measured pre-exposure common causes required for
identification. For each member j and joint state a, identification by
the dyadic g-formula follows under consistency, joint conditional
exchangeability, joint positivity, and no between-dyad interference.
Formally, joint exchangeability requires $Y_j(a)\perp A\mid H$ for every $a\in\mathcal S$. Identification requires
$P(A=a\mid H=h)>0$ wherever $h$ has positive density in the
target population. Stable finite-sample estimation requires the stronger
working condition that relevant joint probabilities be bounded away from
zero on the target support. The latter is not an identification theorem;
it is the practical overlap condition examined in the simulations and
HRS diagnostics.

\begin{equation}
\mu_j(a)=\E\{Y_j(a)\}=\E_H\bigl[\E\{Y_j\mid A=a,H\}\bigr].
\tag{3}
\end{equation}

The use of the joint treatment vector matters. Exchangeability is a
statement about the observed joint state, not separately about each
member in isolation. Likewise, positivity is evaluated for each joint
state. A dataset can have adequate marginal numbers of retired and
nonretired individuals while providing very little support for the
both-treated state 11 in some covariate regions.

\subsection{Longitudinal treatment-confounder
feedback}\label{longitudinal-treatment-confounder-feedback}

With repeated treatment decisions, $H_t$ may contain variables that are
causes of treatment at time t and outcome but are themselves affected by
treatment at time t-1. Conditioning on such a variable in a conventional
outcome regression changes the causal pathway being estimated and can
induce bias for sustained-regime effects. The target trial therefore
distinguishes baseline histories from post-baseline time-varying
confounders and, when repeated treatment strategies are the target, uses
sequential g-method logic rather than automatic covariate adjustment.

For two decision times, the joint treatment mechanism factorizes
sequentially as $g_1(\mathbf A_1\mid H_1)\,g_2(\mathbf A_2\mid H_2,\mathbf A_1)$, where each
$\mathbf A_t$ is itself a dyadic vector. This structure is evaluated in the second
simulation study. The HRS illustration is intentionally simpler: it
defines one joint retirement transition at t+1 and measures CES-D at
t+2, using baseline histories only for confounding adjustment.

\subsection{Time zero, eligibility, and treatment
assignment}\label{time-zero-eligibility-and-treatment-assignment}

Target-trial thinking is most useful when it forces temporal decisions
that a regression specification can leave implicit. Eligibility must be
evaluated using information available before treatment assignment, and
the start of follow-up must be aligned with the treatment strategies
being compared. Hernan and Robins (2016) emphasized that misalignment
can create avoidable design biases. Recent target-trial guidance
likewise treats explicit protocol specification as a way to make the
causal question auditable rather than as a new statistical estimator
(Hernan et al., 2022; Cashin et al., 2025; Simon-Tillaux et al., 2026).

The dyadic setting adds an additional alignment requirement: both
members must belong to the same risk set at the same dyadic time zero.
If one spouse is eligible because of information measured after the
other spouse\textquotesingle s exposure classification, the resulting
pair is not emulating a single coherent assignment event. Likewise, if
eligibility depends on a future property of the relationship, such as
remaining married after treatment, the analysis may condition on a
post-baseline consequence. In the HRS application we therefore define
reciprocal couple membership at baseline and deliberately do not require
the pair to remain together through the outcome wave in the primary
cohort.

Our empirical exposure is classified over the interval from baseline t
to t+1, and CES-D is measured at t+2. This is a pragmatic emulation
rather than a literal randomized retirement trial. The intervention
\textquotesingle fully retired by the next wave\textquotesingle{}
contains versions that cannot be completely standardized in
observational data, including voluntary versus involuntary retirement
and heterogeneous changes in work hours, income, or health insurance.
The consistency assumption therefore requires a treatment-version
interpretation analogous to treatment-variation irrelevance: versions
grouped under the same retirement label must be sufficiently comparable
for the target contrast being claimed (VanderWeele \& Hernán, 2013). The
target-trial formulation makes this limitation visible; it does not
eliminate it.

\subsection{Censoring, death, and dyad
dissolution}\label{censoring-death-and-dyad-dissolution}

Longitudinal dyadic studies face several forms of censoring. A member
can die, miss an interview, enter a nursing facility, or become
otherwise unavailable. The relationship can also dissolve or the partner
can change. These events are scientifically different. Death may
preclude measurement of a later outcome, whereas partner change can
alter the meaning of a partner exposure without necessarily preventing
outcome observation. A rigorous analysis must specify whether such
events define censoring, competing events, treatment versions, or
changes in the target population.

The primary HRS illustration uses a complete observed-outcome cohort and
does not estimate inverse-probability-of-censoring weights. This choice
keeps the application focused on the joint-treatment problem but narrows
the interpretation to dyads with the required observed variables. We
retain post-baseline partner identifiers only as diagnostics.
Ninety-eight percent of the final baseline-defined couples remained with
the same partner at t+1 and approximately 96\% at t+2, so the primary
analysis is not driven by widespread partner turnover. Appendix E
compares the primary estimates with a stable-partner restriction as a
sensitivity analysis.

\section{Estimation}\label{estimation}

\subsection{Why the joint propensity is the default dyadic
object}\label{why-the-joint-propensity-is-the-default-dyadic-object}

For treatment state $a$, define the joint propensity $g_a(H)=P(A=a\mid H)$. A common shortcut is to estimate separate member-specific
propensities and multiply them. For example, $g_{11}(H)$ may be approximated
as $P(A_1=1\mid H)\,P(A_2=1\mid H)$. This factorization is correct
only if $A_1$ and $A_2$ are conditionally independent given $H$. In couples,
residual dependence can remain because of shared plans, common shocks,
negotiation, or omitted dyad-level determinants. The primary simulation
stress-tests this fitted factorized strategy. Because its
separate-propensity arm also estimates parametric marginal models,
however, the resulting bias cannot be attributed to factorization alone;
Section 5.1 states the exact working models and the corresponding
limitation.

\subsection{Joint IPW and AIPW}\label{joint-ipw-and-aipw}

For independent dyads $d=1,\ldots,n$, a joint IPW estimator of the mean
potential outcome for member $j$ under state $a$ is

\begin{equation}
\widehat\mu_{j,\mathrm{IPW}}(a)=\frac1n\sum_{d=1}^n\frac{\mathbf 1(A_d=a)Y_{dj}}{\widehat g_a(H_d)}.
\tag{4}
\end{equation}

Let $m_j(a,H)=\E(Y_j\mid A=a,H)$ denote an outcome regression. The
augmented estimator is
\begin{equation}
\widehat\mu_{j,\mathrm{AIPW}}(a)=\frac1n\sum_{d=1}^n\left[\widehat m_j(a,H_d)+\frac{\mathbf 1(A_d=a)\{Y_{dj}-\widehat m_j(a,H_d)\}}{\widehat g_a(H_d)}\right].
\tag{5}
\end{equation}

The desired direct, spillover, joint, and interaction effects are linear
contrasts of the four estimated regime means. Under standard regularity
conditions, AIPW has the familiar double-robustness property for a fixed
treatment-effect target: consistency can be retained when either the
joint propensity model or the outcome regression is correctly specified
(Bang \& Robins, 2005). Because the primary 36-scenario factorial varied
both nuisance specifications together, Section 5.4 reports a focused
four-cell diagnostic in which the propensity and outcome models are
separately correct or misspecified. Broader partial-interference theory
establishes doubly robust estimation beyond the dyadic special case (Liu
et al., 2019; Qu et al., 2026).

\subsection{Cross-fitting and flexible nuisance
estimation}\label{cross-fitting-and-flexible-nuisance-estimation}

Parametric multinomial and outcome models can fail when confounding is
nonlinear or involves interactions. We therefore consider cross-fitted
AIPW: dyads are partitioned into folds, nuisance models are fit outside
each held-out fold, and propensity and outcome predictions are generated
for the held-out dyads before the influence-function-style
pseudo-outcomes are combined. Cross-fitting reduces overfitting bias
when flexible machine-learning nuisance estimators are used
(Chernozhukov et al., 2018; Zivich \& Breskin, 2021). In the HRS
analysis, folds are stratified by the four treatment states and defined
at the dyad level; random forests estimate the four-category treatment
mechanism and member-specific outcome regressions.

\subsection{Positivity, inference, and
diagnostics}\label{positivity-inference-and-diagnostics}

Doubly robust estimation does not remove the positivity requirement.
Small joint propensities can generate extreme inverse weights and
unstable pseudo-outcomes even when nuisance models are otherwise
correct. We therefore report state counts, lower-tail observed-state
probabilities, fractions below 0.05, inverse-weight diagnostics, and
sensitivity to numerical probability floors. The 0.005 floor used in the
production code is treated as a numerical stabilization choice rather
than as innocuous preprocessing; Section 6 reports untruncated, 0.005,
0.010, and 0.020 sensitivity for the parametric joint estimators.

For the empirical IPW/AIPW estimators, standard errors are calculated
from the empirical standard deviation of dyad-level contrast
pseudo-outcomes divided by the square root of the number of dyads. For
plain joint IPW this fixed-propensity-style variance calculation is
conservative when the propensity is estimated; the primary
simulation\textquotesingle s 0.991 and 0.977 mean coverage in the two
linear conditions should therefore be read as over-coverage, not
superior inference. The APIM-style benchmark uses uncertainty from 2,000
nonparametric dyad bootstrap resamples so that the associational
comparator is not presented without inference.

\subsection{Why multiplying separate propensities can be wrong}\label{why-multiplying-separate-propensities-can-be-wrong}

Let $p_1(H)=P(A_1=1\mid H)$ and $p_2(H)=P(A_2=1\mid H)$, and let $q_a(H)$ denote the four-state distribution obtained by multiplying these true marginals:
\[
q_{11}=p_1p_2,\qquad q_{10}=p_1(1-p_2),\qquad q_{01}=(1-p_1)p_2,\qquad q_{00}=(1-p_1)(1-p_2).
\]
Define the conditional dependence discrepancy $\delta(H)=g_{11}(H)-p_1(H)p_2(H)$. For a binary dyad, $\delta(H)=\Cov(A_1,A_2\mid H)$, and fixed margins force every discrepancy between the true joint law $g$ and the factorized law $q$ into one direction.

\[
(g_{00},g_{10},g_{01},g_{11})=(q_{00}+\delta,\ q_{10}-\delta,\ q_{01}-\delta,\ q_{11}+\delta).
\]

\begin{proposition}[Exact factorization geometry and IPW bias]\label{prop:factorization}
Assume $q_a(H)>0$ for the states used by the target. With $s=(+1,-1,-1,+1)$ in cell order $(00,10,01,11)$,
\[
g_a(H)=q_a(H)+s_a\delta(H).
\]
Hence the conditional total-variation distance is
\[
\TV\{g(\cdot\mid H),q(\cdot\mid H)\}=2|\delta(H)|,
\qquad \TV(g,q)\equiv\tfrac12\lVert g-q\rVert_1.
\]
If $\mu_a=\E\{m_a(H)\}$ and a Horvitz--Thompson regime mean is instead formed with the factorized denominator $q_a$, its population bias is exactly
\[
\widetilde\mu_a-\mu_a=\E\!\left\{\frac{s_a\delta(H)m_a(H)}{q_a(H)}\right\}.
\]
For $C=\sum_a c_a\mu_a$,
\[
\widetilde C-C=\E\!\left[\delta(H)\sum_a\frac{c_as_am_a(H)}{q_a(H)}\right].
\]
\end{proposition}

The proposition separates two distinct questions. Marginal-model accuracy concerns whether $p_1(H)$ and $p_2(H)$ are estimated correctly. Factorization concerns whether their product is the joint law. Even exact knowledge of both marginals leaves bias whenever $\delta(H)\neq0$ and the outcome-weighted terms in Proposition~\ref{prop:factorization} do not cancel. The total-variation identity also gives $\delta$ a direct probability-scale interpretation. The algebra is elementary for a $2\times2$ table; its role here is diagnostic, making clear exactly what the additional conditional-independence assumption contributes to a causal weighting analysis. Proof is in Appendix G.2. The displayed bias identity is specific to the unnormalized Horvitz--Thompson functional used here; normalized or stabilized H\'ajek weighting introduces an additional denominator effect, so the formula should not be applied verbatim to those estimators.

\begin{corollary}[AIPW protection against factorization error]\label{cor:aipw-factorization}
Let $\bar g_a$ and $\bar m_a$ denote probability limits of working propensity and outcome regressions. The population AIPW functional for regime $a$ differs from $\mu_a$ by
\[
\E\!\left[\left(1-\frac{g_a}{\bar g_a}\right)(\bar m_a-m_a)\right].
\]
If the working propensity is the true-marginal factorization $q_a$, then
\[
\Psi_{a,\mathrm{AIPW}}-\mu_a=-\E\!\left[\frac{s_a\delta(H)\{\bar m_a(H)-m_a(H)\}}{q_a(H)}\right].
\]
\end{corollary}

Thus treatment dependence alone does not invalidate an outcome-regression strategy, and factorization error alone does not bias AIPW when the outcome regression is correct. Bias requires the treatment-law discrepancy to interact with outcome-model error. Conversely, a correctly specified joint propensity protects against outcome misspecification. This identity is the population version of the double-robustness pattern examined in Section 5.4. In practice, any coherent probabilistic learner may estimate the four-state joint law; multinomial regression is only one option. Proof is in Appendix G.3.

\subsection{Contrast-level influence functions and
implementation}\label{contrast-level-influence-functions-and-implementation}

Equation (5) yields one pseudo-outcome $\phi_j(a)$ for each member $j$ and state $a$.
Because the actor, spillover, joint, and interaction effects are linear
contrasts of regime means, their influence-function-style contributions
are obtained by taking the same linear contrasts of the state-specific
pseudo-outcomes. For example, the member-1 direct contribution is
$\phi_1(10)-\phi_1(00)$, and the interaction contribution is
$\phi_1(11)-\phi_1(10)-\phi_1(01)+\phi_1(00)$. The symmetric dyadic effect
averages the corresponding member-1 and member-2 contributions. Standard
errors can therefore be calculated from the empirical variance of the
dyad-level contrast contributions, treating dyads as the independent
sampling units.

The empirical implementation follows five steps. First, define the
analysis cohort and four-level joint state without using post-treatment
information. Second, partition dyads into cross-fitting folds,
stratifying by state when feasible so that rare treatment cells are
represented across folds. Third, fit the joint propensity and two
member-specific outcome regressions on the training folds. Fourth,
predict all four treatment-state probabilities and all counterfactual
outcome means for held-out dyads. Fifth, construct held-out AIPW
pseudo-outcomes and combine them only after every dyad has out-of-fold
nuisance predictions. This sequencing keeps the machine-learning step
separate from the final causal contrast. This implementation logic
follows the broader partial-interference software and workflow tradition
exemplified by Saul and Hudgens (2017), while using dyads as the
independent sampling units.

Cross-fitting does not guarantee good finite-sample behavior when the
treatment cell is extremely rare. AIPW contains the term $\mathbf 1(A=a)/g_a(H)$,
so an observation in a poorly supported state can still have a large
contribution. For that reason, nuisance-model diagnostics should be
paired with overlap diagnostics: treatment-state counts, minimum and
lower-tail observed-state probabilities, weight distributions, and, when
possible, plots of state-specific propensity distributions. A small
prediction error metric for the propensity model is not a substitute for
causal support.

\section{Simulation Studies}\label{simulation-studies}

\subsection{Primary simulation: joint treatment dependence and
nuisance
misspecification}\label{primary-simulation-joint-treatment-dependence-and-nuisance-misspecification}

The primary simulation generated independent dyads with two binary
treatments and continuous member-specific outcomes. Baseline covariates
U, X1, and X2 were standard normal. Joint treatment probabilities
followed a four-category softmax with η00=0, η10=λ1, η01=λ2, and
η11=λ1+λ2+κ. Thus κ is the conditional log odds ratio between A1 and A2
given the simulated history: κ=0 implies conditional independence,
whereas larger κ induces residual dyadic treatment dependence. Member-1
outcomes had true direct, spillover, and interaction coefficients 0.50,
0.30, and 0.25, yielding a joint contrast of 1.05; member 2 followed a
symmetric construction. Separate-PS IPW fit two main-effects Bernoulli
logistic models, A1 \textasciitilde{} U + X1 + X2 and A2
\textasciitilde{} U + X1 + X2, and multiplied their fitted
probabilities. Under the softmax law, for example, logit
P(A1=1\textbar H) = λ1 + log\{{[}1+exp(λ2+κ){]}/{[}1+exp(λ2){]}\}; when
κ is nonzero this marginal is generally not linear-logistic in H even
when λ1 and λ2 are. Accordingly, the fitted separate-PS arm can reflect
both factorization error and marginal working-model misspecification.

The full factorial design varied dyad sample size (500, 1,000, 3,000), κ
(0, 0.8, 1.6), confounding form (linear or nonlinear), and positivity
(good or near), producing 36 scenarios with 500 Monte Carlo replications
each. Estimators were APIM-style outcome regression, known-propensity
joint IPW, separate-propensity IPW, joint multinomial IPW, joint
parametric AIPW, and cross-fitted random-forest AIPW. For computational
tractability the full simulation used 2-fold cross-fitting and 300
trees, whereas the HRS illustration used 5 folds and 500 trees; the
estimand and estimating equation were unchanged. Estimated probabilities
were protected by a 0.005 numerical floor. One dataset-level numerical
failure occurred in the n=500, κ=0, linear/good scenario; all other
primary scenarios completed all 500 replications.

Table 3. Primary simulation performance averaged across sample sizes,
dependence levels, and four causal contrasts

\begin{longtable}[]{@{}
  >{\raggedright\arraybackslash}p{(\columnwidth - 8\tabcolsep) * \real{0.2000}}
  >{\raggedright\arraybackslash}p{(\columnwidth - 8\tabcolsep) * \real{0.2000}}
  >{\raggedright\arraybackslash}p{(\columnwidth - 8\tabcolsep) * \real{0.2000}}
  >{\raggedright\arraybackslash}p{(\columnwidth - 8\tabcolsep) * \real{0.2000}}
  >{\raggedright\arraybackslash}p{(\columnwidth - 8\tabcolsep) * \real{0.2000}}@{}}
\toprule\noalign{}
\begin{minipage}[b]{\linewidth}\raggedright
\textbf{Condition}
\end{minipage} & \begin{minipage}[b]{\linewidth}\raggedright
\textbf{Estimator}
\end{minipage} & \begin{minipage}[b]{\linewidth}\raggedright
\textbf{Mean \textbar bias|}
\end{minipage} & \begin{minipage}[b]{\linewidth}\raggedright
\textbf{Mean RMSE}
\end{minipage} & \begin{minipage}[b]{\linewidth}\raggedright
\textbf{Mean coverage}
\end{minipage} \\
\midrule\noalign{}
\endhead
\bottomrule\noalign{}
\endlastfoot
Linear / good & APIM-style OR & 0.004 & 0.108 & 0.945 \\
Linear / good & Separate-PS IPW & 0.286 & 0.402 & 0.746 \\
Linear / good & Joint multinomial IPW & 0.007 & 0.154 & 0.991 \\
Linear / good & Joint parametric AIPW & 0.006 & 0.126 & 0.943 \\
Linear / good & Cross-fitted RF-AIPW & 0.038 & 0.236 & 0.959 \\
Linear / good & Known-propensity joint IPW & 0.006 & 0.200 & 0.948 \\
Linear / near & APIM-style OR & 0.004 & 0.110 & 0.949 \\
Linear / near & Separate-PS IPW & 0.262 & 0.474 & 0.850 \\
Linear / near & Joint multinomial IPW & 0.009 & 0.270 & 0.977 \\
Linear / near & Joint parametric AIPW & 0.006 & 0.177 & 0.948 \\
Linear / near & Cross-fitted RF-AIPW & 0.022 & 0.282 & 0.948 \\
Linear / near & Known-propensity joint IPW & 0.017 & 0.356 & 0.944 \\
Nonlinear / good & APIM-style OR & 0.339 & 0.398 & 0.486 \\
Nonlinear / good & Separate-PS IPW & 0.475 & 0.563 & 0.479 \\
Nonlinear / good & Joint multinomial IPW & 0.362 & 0.480 & 0.553 \\
Nonlinear / good & Joint parametric AIPW & 0.377 & 0.481 & 0.514 \\
Nonlinear / good & Cross-fitted RF-AIPW & 0.048 & 0.307 & 0.924 \\
Nonlinear / good & Known-propensity joint IPW & 0.031 & 0.874 & 0.921 \\
Nonlinear / near & APIM-style OR & 0.495 & 0.542 & 0.427 \\
Nonlinear / near & Separate-PS IPW & 0.820 & 0.969 & 0.416 \\
Nonlinear / near & Joint multinomial IPW & 0.804 & 0.997 & 0.512 \\
Nonlinear / near & Joint parametric AIPW & 0.797 & 0.985 & 0.508 \\
Nonlinear / near & Cross-fitted RF-AIPW & 0.166 & 0.377 & 0.773 \\
Nonlinear / near & Known-propensity joint IPW & 0.095 & 2.602 & 0.879 \\
\end{longtable}

Note. Known-propensity joint IPW uses the known data-generating
treatment probabilities and is included only as a calibration benchmark.
APIM-style OR is correctly specified in the linear outcome conditions.
Coverage for IPW uses the pseudo-outcome variance that treats fitted
propensities as fixed, which is conservative in the correctly specified
linear settings.

\subsection{True-marginal and fitted-logit factorization diagnostics}\label{true-marginal-and-fitted-logit-factorization-diagnostics}

Proposition~\ref{prop:factorization} suggests a direct diagnostic for the simulation criticism that fitted separate propensities can mix factorization error with marginal-model misspecification. We therefore performed two population diagnostics under linear confounding and good positivity. First, we computed the exact DGP marginals $p_1(H)$ and $p_2(H)$, multiplied them to form $q_a(H)$, and evaluated the resulting Horvitz--Thompson contrasts by Monte Carlo integration over 2,000,000 independent draws of the baseline covariates. No treatment model is estimated in this true-marginal calculation, so any nonzero bias is caused by factorization alone. Second, to determine whether the main-effects Bernoulli working models materially change the contrast pattern, we solved their population score equations, multiplied the resulting fitted marginal probabilities, applied the same 0.005 probability floor and renormalization used in the finite-sample Separate-PS arm, and evaluated the contrasts on an independent 5,000,000-draw Monte Carlo sample. This fitted-logit population diagnostic adds marginal functional-form error while essentially removing finite-sample estimation noise. Batch Monte Carlo standard errors for the reported mean absolute biases were below 0.001.

\begin{table}[H]
\centering
\captionsetup{labelformat=empty}
\caption{\textbf{Table 3A.} True-marginal and fitted-logit factorization diagnostics under linear confounding and good positivity}
\label{tab:true-marginal-factorization}
\footnotesize
\setlength{\tabcolsep}{3.5pt}
\begin{tabular}{c c c c c}
\toprule
$\kappa$ & \shortstack{True-marginal factorized biases\\$(D,S,J,I)$} & \shortstack{Fitted-logit biases\\$(D,S,J,I)$} & \shortstack{True-marginal\\mean $|\mathrm{bias}|$} & \shortstack{Fitted simulation\\mean $|\mathrm{bias}|$} \\
\midrule
0.0 & (0.000, 0.000, 0.000, 0.000) & (0.000, 0.000, 0.000, 0.000) & 0.000 & 0.004 \\
0.8 & ($-$0.200, $-$0.277, $+$0.027, $+$0.503) & ($-$0.200, $-$0.278, $+$0.026, $+$0.504) & 0.251 & 0.259 \\
1.6 & ($-$0.545, $-$0.647, $-$0.252, $+$0.940) & ($-$0.543, $-$0.649, $-$0.253, $+$0.938) & 0.596 & 0.596 \\
\bottomrule
\end{tabular}
\end{table}

\noindent\textit{Note.} True-marginal values use the true member-specific marginal treatment probabilities from the DGP and therefore remove marginal propensity estimation and working-model error while retaining the false conditional-independence restriction. Fitted-logit values solve the population score equations for the two main-effects Bernoulli propensity working models and then apply the same 0.005 probability floor and renormalization used by the simulation. The final column is the finite-sample Separate-PS mean absolute bias averaged across sample sizes, as reported in Appendix D.1. The component-wise agreement between the true-marginal and fitted-logit bias vectors shows that the matching mean absolute biases are not hiding offsetting errors across contrasts; the dominant mechanism in this DGP is the false factorization itself.

The contrast-specific pattern is more consequential than the scalar mean absolute bias. At $\kappa=1.6$, the true-marginal factorized biases imply targets of approximately $-0.045$ for the direct effect (truth 0.50), $-0.347$ for the spillover effect (truth 0.30), 0.798 for the $11$-versus-$00$ joint effect (truth 1.05), and 1.190 for the additive interaction (truth 0.25). Thus factorizing even the true marginals essentially nullifies the direct contrast and moves it slightly below zero, reverses the sign of the spillover contrast, attenuates the joint contrast by about 24\%, and inflates the interaction to about 4.8 times its true value. The fitted-logit population calculation produces nearly the same contrast-specific pattern. Proposition~\ref{prop:factorization} explains why no single scalar dependence measure can be translated into one universal bias magnitude: $\delta(H)$ is filtered through state-specific support $q_a(H)$ and conditional outcome means $m_a(H)$.

\subsection{Primary simulation
results}\label{primary-simulation-results}

The correctly specified linear scenarios are the most important calibration for the paper's framing. APIM-style outcome regression was essentially unbiased and had the lowest mean RMSE of the implementable estimators in both linear conditions (0.108 under good positivity and 0.110 under near positivity). Residual treatment dependence did not invalidate a correctly specified outcome model. By contrast, fitted separate-PS IPW mean absolute bias increased from 0.004 at $\kappa=0$ to 0.259 at $\kappa=0.8$ and 0.596 at $\kappa=1.6$ under linear confounding and good positivity, while joint multinomial IPW and joint parametric AIPW remained near unbiased over the same dependence gradient. Table~3A shows that this pattern is not mainly an artifact of the fitted marginal logits: multiplying the exact DGP marginals yields population mean absolute biases of 0.251 and 0.596 at $\kappa=0.8$ and 1.6, and the fitted-logit population bias vectors track the true-marginal bias vectors component by component. The agreement therefore cannot be explained by offsetting contrast-specific errors hidden by a mean absolute bias summary. The simulation verifies the factorization mechanism isolated by Proposition~\ref{prop:factorization} while retaining the broader practical comparison among estimators.

Nonlinear confounding produced a different pattern because the
parametric outcome and treatment regressions were both misspecified.
Under good positivity, APIM-style OR, joint multinomial IPW, and
parametric AIPW had similar mean absolute bias (0.339, 0.362, and
0.377), while cross-fitted RF-AIPW reduced mean absolute bias to 0.048.
This does not show that joint treatment modeling alone solves
nonlinearity: the gain came from flexible nuisance learning. RF-AIPW
bias decreased from approximately 0.073 at n=500 to 0.046 at n=1,000 and
0.024 at n=3,000.

Flexible nuisance estimation did not solve weak support. Under nonlinear
confounding with near positivity, RF-AIPW retained the smallest mean
absolute bias among the methods using estimated nuisance functions
(0.166), but mean coverage fell to 0.773. Known-propensity joint IPW
also remained low-bias while its mean RMSE rose to 2.602, demonstrating
that knowledge of the true treatment mechanism cannot create information
in poorly supported joint states. The realized support diagnostics make
the calibration concrete. Averaged across replications, about 0.38\% and
1.64\% of observed states had true propensity below 0.05 in the
linear/good and linear/near conditions, and about 0.80\% and 1.59\% did
so in the nonlinear/good and nonlinear/near conditions. The median
within-replication minimum observed-state propensity was approximately
0.0327, 0.0110, 0.0136, and 0.0067, respectively. The HRS data are at
least as challenging in the lower tail: 3.3\% of observed states had
untruncated parametric propensity below 0.05, with a minimum of 0.0019.

\begin{figure}[H]
\centering
\includegraphics[width=0.88\textwidth]{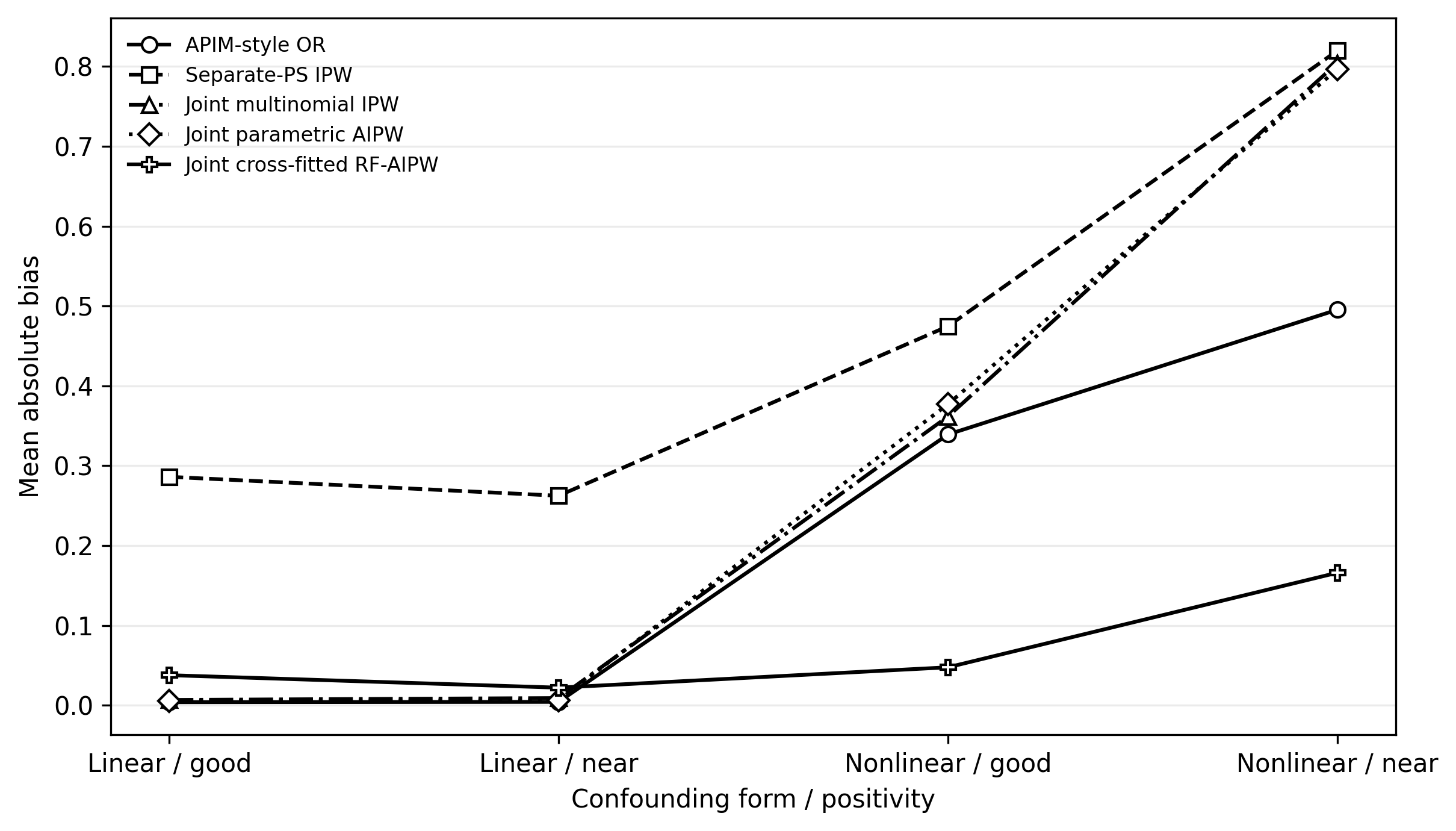}
\caption{Primary simulation mean absolute bias. Under correct linear outcome specification, APIM-style regression remains low-bias despite treatment dependence; the fitted separate-propensity weighting strategy deteriorates as $\kappa$ increases. Under nonlinear confounding, flexible nuisance learning rather than joint modeling by itself drives the improvement.}
\end{figure}

\subsection{Focused double-robustness
diagnostic}\label{focused-double-robustness-diagnostic}

The 36-scenario factorial changed the functional form of the treatment
and outcome mechanisms together, so it did not by itself isolate the
usual double-robustness claim. We conducted a focused 500-replication
diagnostic at n=1,000, κ=1.6, and good positivity in which linear versus
nonlinear treatment assignment and linear versus nonlinear outcome
surfaces were crossed separately. Parametric AIPW remained near unbiased
when either the multinomial propensity or the outcome regression was
correctly specified, and failed when both were misspecified. The
slightly smaller RMSE when the propensity model was misspecified but the
outcome model was correct (0.123 versus 0.136 when both were correct)
should not be interpreted as an efficiency gain from misspecification;
with 500 Monte Carlo replications and a correctly specified outcome
regression dominating the augmentation, small RMSE reversals of this
size can occur. This focused diagnostic is a verification of the
standard AIPW property, not a new interference result.

Table 4. Focused double-robustness diagnostic for joint parametric AIPW

\begin{longtable}[]{@{}
  >{\raggedright\arraybackslash}p{(\columnwidth - 6\tabcolsep) * \real{0.2500}}
  >{\raggedright\arraybackslash}p{(\columnwidth - 6\tabcolsep) * \real{0.2500}}
  >{\raggedright\arraybackslash}p{(\columnwidth - 6\tabcolsep) * \real{0.2500}}
  >{\raggedright\arraybackslash}p{(\columnwidth - 6\tabcolsep) * \real{0.2500}}@{}}
\toprule\noalign{}
\begin{minipage}[b]{\linewidth}\raggedright
\textbf{Nuisance specification}
\end{minipage} & \begin{minipage}[b]{\linewidth}\raggedright
\textbf{Mean \textbar bias|}
\end{minipage} & \begin{minipage}[b]{\linewidth}\raggedright
\textbf{Mean RMSE}
\end{minipage} & \begin{minipage}[b]{\linewidth}\raggedright
\textbf{Mean coverage}
\end{minipage} \\
\midrule\noalign{}
\endhead
\bottomrule\noalign{}
\endlastfoot
Both nuisance models correct & 0.006 & 0.136 & 0.948 \\
Propensity correct; outcome misspecified & 0.006 & 0.229 & 0.949 \\
Propensity misspecified; outcome correct & 0.008 & 0.123 & 0.951 \\
Both nuisance models misspecified & 0.308 & 0.417 & 0.507 \\
\end{longtable}

Note. 500 replications per cell; n=1,000; κ=1.6; good positivity.
Averages are over direct, spillover, joint, and interaction contrasts.

\subsection{Longitudinal extension}\label{longitudinal-extension}

A second 36-scenario simulation examined treatment-confounder feedback
over two dyadic treatment decisions. Because the HRS application in this
paper defines a single joint retirement transition, that simulation
addresses a different longitudinal scientific target. Its design and
results are therefore reported in Appendix A. The extension illustrates
why conventional adjustment for treatment-induced covariates can fail
for sustained dyadic regimes, but it is not used as evidence for the
point-treatment HRS application.

\section{Empirical Illustration: Retirement Transitions and
Depressive Symptoms in the
HRS}\label{empirical-illustration-retirement-transitions-and-depressive-symptoms-in-the-hrs}

\subsection{Data and cohort
construction}\label{data-and-cohort-construction}

We used the RAND HRS Longitudinal File 2022 (Version 1), which
harmonizes Health and Retirement Study data from 1992 through 2022. The
HRS is a nationally prominent longitudinal study of U.S. adults over age
50 and their spouses or partners (Sonnega et al., 2014). It is sponsored
by the National Institute on Aging (U01AG009740 and R01AG073289) and
conducted by the University of Michigan. The RAND longitudinal product
used here was released in May 2025. This is a secondary analysis of
deidentified public-use data.

The empirical example was selected to illustrate joint treatment
assignment in a naturally dyadic life transition without duplicating the
simulation DGP. Baseline candidates began at wave 4 (1998) and extended
through wave 14 (2018), allowing an exposure classification at t+1 and
an outcome at t+2. A reciprocal couple link was required at baseline.
Both partners had to be age 50-69 and working full- or part-time at
baseline. At t+1, each member had to be either still working
full/part-time or fully retired; unemployment, disability,
partial-retirement, and other labor-force states were excluded from this
specific binary contrast. The treatment vector was coded 00 if neither
partner was fully retired at t+1, 10 or 01 if one member retired, and 11
if both retired. The analytic cohort also required observed
member-specific RAND CES-D at t+2.

To avoid one additional source of post-baseline selection, the primary
cohort did not require couples to remain partnered after baseline.
Future spouse links were retained only as diagnostics. Among the final
couples, 98.1\% had the same baseline partner at t+1 and 96.5\% at t+2;
2,311 of 2,400 couples remained with the same partner through t+2. When
a couple met eligibility at more than one baseline, only the first
eligible baseline was retained, so the analytic rows represent unique
dyads rather than repeated risk-set entries.

\begin{figure}[H]
\centering
\includegraphics[width=0.88\textwidth]{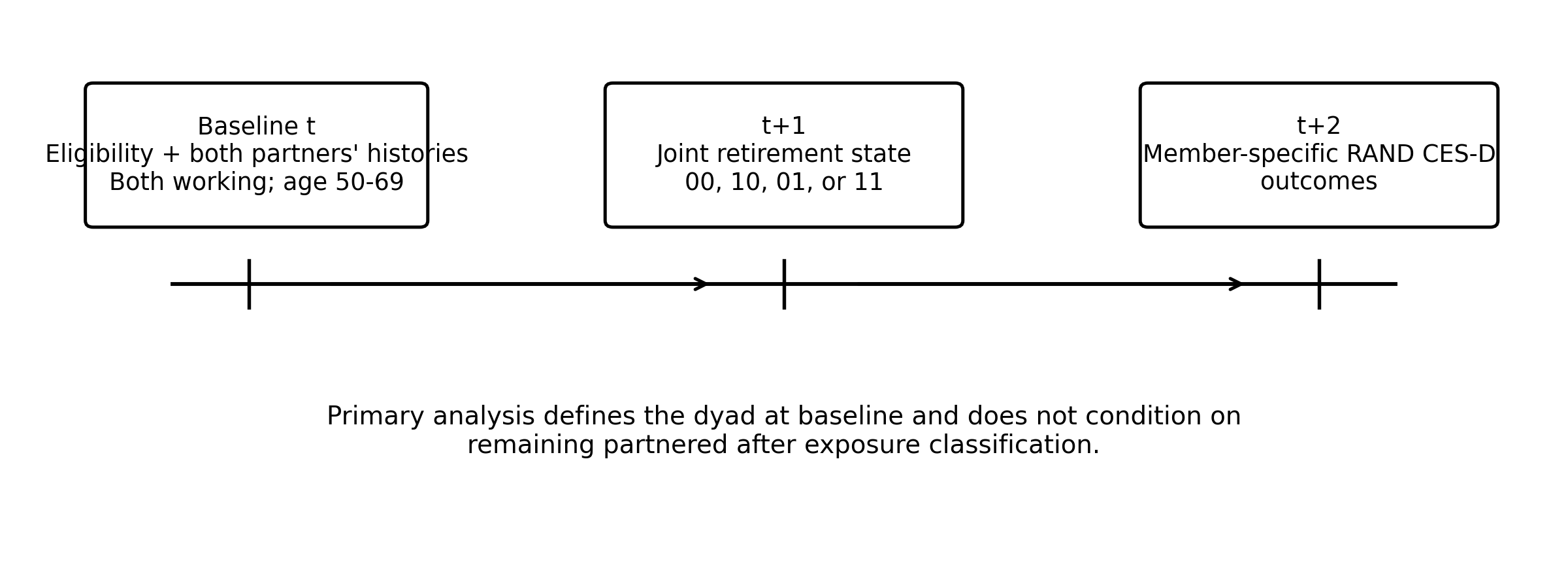}
\caption{HRS target-trial timeline. Baseline defines eligibility and the confounding history; the joint retirement state is classified one wave later, and member-specific RAND CES-D outcomes are measured at the following wave. Post-baseline partnership survival is not an eligibility condition in the primary analysis.}
\end{figure}

\subsection{Covariates and
estimators}\label{covariates-and-estimators}

The baseline adjustment set included both partners\textquotesingle{}
age, RAND CES-D, chronic-condition burden, years of education, gender,
race, Hispanic ethnicity, baseline full-time-work indicator, and
baseline-wave fixed effects. The analysis used no post-baseline
covariates for confounding adjustment. Member 1 was defined
administratively as the partner with the lower HHIDPN and member 2 as
the higher HHIDPN; substantive effects were reported as symmetric
averages so that this arbitrary ordering does not determine the causal
contrast.

We fit five estimators as complementary diagnostics rather than
designating one as empirically preferred. APIM-style adjusted outcome
regression included both retirement indicators, their interaction, and
baseline covariates. Separate-PS IPW multiplied two Bernoulli
propensities and therefore imposed conditional treatment independence.
Joint multinomial IPW modeled the four-state treatment directly. Joint
parametric AIPW combined multinomial propensity and linear outcome
models. Cross-fitted random-forest AIPW used 5 state-stratified folds
and 500 trees for each nuisance learner. The production analysis used a
0.005 probability floor before renormalization; Section 6.6 reports
sensitivity because that numerical choice is binding for rare-state
contrasts.

Table 5. HRS joint retirement-state distribution

\begin{longtable}[]{@{}
  >{\raggedright\arraybackslash}p{(\columnwidth - 4\tabcolsep) * \real{0.3333}}
  >{\raggedright\arraybackslash}p{(\columnwidth - 4\tabcolsep) * \real{0.3333}}
  >{\raggedright\arraybackslash}p{(\columnwidth - 4\tabcolsep) * \real{0.3333}}@{}}
\toprule\noalign{}
\begin{minipage}[b]{\linewidth}\raggedright
\textbf{Joint state}
\end{minipage} & \begin{minipage}[b]{\linewidth}\raggedright
\textbf{Couples}
\end{minipage} & \begin{minipage}[b]{\linewidth}\raggedright
\textbf{Percent}
\end{minipage} \\
\midrule\noalign{}
\endhead
\bottomrule\noalign{}
\endlastfoot
00: neither retires & 2,135 & 89.0\% \\
10: member 1 retires & 137 & 5.7\% \\
01: member 2 retires & 88 & 3.7\% \\
11: both retire & 40 & 1.7\% \\
\end{longtable}

\emph{Note.} N=2,400 unique couples. Member labels are administrative;
causal estimates average symmetrically over both orientations.

\subsection{Positivity and treatment-state
support}\label{positivity-and-treatment-state-support}

The joint-state distribution was highly unbalanced: 89.0\% of couples
remained in 00, whereas only 40 couples (1.7\%) entered 11. Before
probability flooring, the parametric joint model had minimum
observed-state propensity 0.0019, first percentile 0.0195, 3.33\% of
observed states below 0.05, and maximum inverse weight 532. After the
0.005 floor and renormalization, the maximum weight was approximately
201. The cross-fitted RF model showed the same practical problem after
flooring, with about 3.7\% below 0.05. These lower-tail diagnostics are
at least as severe as the primary simulation\textquotesingle s
near-positivity conditions, where approximately 1.6\% of realized states
had true propensity below 0.05 on average. Thus the empirical
illustration lies in the operating regime where the simulation showed
degraded RF-AIPW coverage.

Residual dyadic treatment dependence was also empirically substantial.
We fit a constrained joint-treatment model with the same baseline
covariates in both member-specific linear predictors and a constant
association parameter κ in the 11 log-linear term. Under this
parameterization κ is the conditional log odds ratio between the two
retirement indicators, exactly matching the dependence parameter used in
the linear simulation DGP. The estimate was κ=1.69 (SE 0.25),
corresponding to a conditional odds ratio of 5.43 (95\% CI 3.32 to
8.87). A one-degree-of-freedom likelihood-ratio comparison with κ=0 gave
χ²=41.40, p\textless0.001. The HRS dependence magnitude is therefore
close to, and slightly above, the simulation\textquotesingle s strongest
κ=1.6 condition. This anchors the separate-PS stress test in an observed
dyadic setting, although the simulation bias magnitude should not be
transported numerically to HRS.

Table 6. Conditional treatment-dependence diagnostic in the HRS cohort

\begin{longtable}[]{@{}
  >{\raggedright\arraybackslash}p{(\columnwidth - 6\tabcolsep) * \real{0.2500}}
  >{\raggedright\arraybackslash}p{(\columnwidth - 6\tabcolsep) * \real{0.2500}}
  >{\raggedright\arraybackslash}p{(\columnwidth - 6\tabcolsep) * \real{0.2500}}
  >{\raggedright\arraybackslash}p{(\columnwidth - 6\tabcolsep) * \real{0.2500}}@{}}
\toprule\noalign{}
\begin{minipage}[b]{\linewidth}\raggedright
\textbf{Quantity}
\end{minipage} & \begin{minipage}[b]{\linewidth}\raggedright
\textbf{Estimate}
\end{minipage} & \begin{minipage}[b]{\linewidth}\raggedright
\textbf{SE / df}
\end{minipage} & \begin{minipage}[b]{\linewidth}\raggedright
\textbf{95\% CI / p}
\end{minipage} \\
\midrule\noalign{}
\endhead
\bottomrule\noalign{}
\endlastfoot
Residual log-odds ratio κ & 1.692 & 0.251 & {[}1.200, 2.183{]} \\
Conditional odds ratio exp(κ) & 5.43 & -\/- & {[}3.32, 8.87{]} \\
LR test of κ=0 & χ²=41.40 & df=1 & p\textless0.001 \\
κ × centered absolute age gap & -0.113 & 0.074 & {[}-0.257, 0.032{]};
p=0.126 \\
\end{longtable}

Note. The constrained model uses the same baseline covariates as the
treatment models and parameterizes η11=η10+η01+κ. Thus κ is directly
comparable to the primary simulation\textquotesingle s
residual-dependence parameter. Orientation, calendar-period, and age-gap
interactions are sensitivity checks; none implies that κ is homogeneous
over the full covariate history.

We also checked whether the constant-κ summary was being driven by
member ordering, calendar period, or a simple measured feature of the
dyad. Orientation-specific conditional logistic models gave log
odds-ratio estimates of 1.698 (SE 0.248; 95\% CI 1.211 to 2.184) for A1
conditional on A2 and 1.720 (SE 0.251; 95\% CI 1.227 to 2.213) for A2
conditional on A1, indicating little orientation asymmetry. Allowing the
dependence coefficient to differ before versus after 2010 gave κ=1.603
(SE 0.283) in earlier baselines and κ=1.985 (SE 0.473) in later
baselines; the early-versus-late interaction was 0.382 (SE 0.537,
p=0.477). We then allowed the partner-treatment coefficient to vary with
centered absolute age gap. The κ-by-age-gap interaction was -0.113 per
year (SE 0.074; 95\% CI -0.257 to 0.032; p=0.126). The implied
conditional log odds ratio was about 2.00 at a one-year age gap and 1.55
at a five-year gap. These checks still do not establish homogeneous
dependence in the full covariate history H, but they provide a direct
observed-H heterogeneity check in addition to orientation and
calendar-time sensitivity. We therefore treat κ=1.69 as a working
summary, not a structural constant. The age-gap interaction is the
manuscript\textquotesingle s direct measured-H heterogeneity check; its
uncertainty is wide enough that meaningful effect modification cannot be
ruled out, so the constant-κ simulation should be read as a controlled
calibration rather than a literal model of the HRS treatment mechanism.

The dependence estimate is a diagnostic, not a bias coefficient. Under
an interior 2×2 table, a conditional odds ratio above one implies
positive conditional covariance and therefore δ(H)$>0$, so the
fitted HRS treatment mechanism is incompatible with exact factorization
into its two member-specific marginals. Proposition 2 nevertheless shows
why κ=1.69 cannot be converted into a numerical causal-estimator bias:
the bias also depends on the history-specific factorized support qa(H)
and the state-specific conditional outcome means. This is why the
simulation\textquotesingle s separate-PS bias is used only to calibrate
a failure mode, not transported to the HRS estimates.

\subsection{Effect estimates}\label{effect-estimates}

Table 7. HRS estimates of retirement-transition contrasts on subsequent
CES-D

\begin{longtable}[]{@{}
  >{\raggedright\arraybackslash}p{(\columnwidth - 8\tabcolsep) * \real{0.2000}}
  >{\raggedright\arraybackslash}p{(\columnwidth - 8\tabcolsep) * \real{0.2000}}
  >{\raggedright\arraybackslash}p{(\columnwidth - 8\tabcolsep) * \real{0.2000}}
  >{\raggedright\arraybackslash}p{(\columnwidth - 8\tabcolsep) * \real{0.2000}}
  >{\raggedright\arraybackslash}p{(\columnwidth - 8\tabcolsep) * \real{0.2000}}@{}}
\toprule\noalign{}
\begin{minipage}[b]{\linewidth}\raggedright
\textbf{Estimator}
\end{minipage} & \begin{minipage}[b]{\linewidth}\raggedright
\textbf{Estimand}
\end{minipage} & \begin{minipage}[b]{\linewidth}\raggedright
\textbf{Estimate}
\end{minipage} & \begin{minipage}[b]{\linewidth}\raggedright
\textbf{95\% CI lower}
\end{minipage} & \begin{minipage}[b]{\linewidth}\raggedright
\textbf{95\% CI upper}
\end{minipage} \\
\midrule\noalign{}
\endhead
\bottomrule\noalign{}
\endlastfoot
APIM-style adjusted OR & Actor & 0.192 & -0.072 & 0.451 \\
Separate-PS IPW & Actor & -0.019 & -0.375 & 0.338 \\
Joint multinomial IPW & Actor & 0.202 & -0.218 & 0.621 \\
Joint parametric AIPW & Actor & 0.241 & -0.055 & 0.537 \\
Cross-fitted RF-AIPW & Actor & 0.141 & -0.316 & 0.599 \\
APIM-style adjusted OR & Spillover & -0.039 & -0.225 & 0.167 \\
Separate-PS IPW & Spillover & -0.268 & -0.576 & 0.040 \\
Joint multinomial IPW & Spillover & -0.109 & -0.472 & 0.254 \\
Joint parametric AIPW & Spillover & 0.023 & -0.255 & 0.302 \\
Cross-fitted RF-AIPW & Spillover & 0.309 & -0.280 & 0.898 \\
APIM-style adjusted OR & Joint & 0.336 & -0.066 & 0.778 \\
Separate-PS IPW & Joint & 0.649 & -0.282 & 1.579 \\
Joint multinomial IPW & Joint & -0.325 & -0.779 & 0.128 \\
Joint parametric AIPW & Joint & 0.150 & -0.080 & 0.379 \\
Cross-fitted RF-AIPW & Joint & -0.069 & -0.481 & 0.342 \\
APIM-style adjusted OR & Interaction & 0.183 & -0.341 & 0.716 \\
Separate-PS IPW & Interaction & 0.936 & -0.141 & 2.012 \\
Joint multinomial IPW & Interaction & -0.418 & -1.218 & 0.381 \\
Joint parametric AIPW & Interaction & -0.115 & -0.632 & 0.403 \\
Cross-fitted RF-AIPW & Interaction & -0.520 & -1.436 & 0.397 \\
\end{longtable}

Note. Positive values indicate higher subsequent CES-D. APIM-style
intervals are percentile intervals from 2,000 dyad bootstrap resamples.
IPW/AIPW intervals use dyad-level pseudo-outcome Wald standard errors.
Simulation results show that the RF-AIPW Wald interval can undercover
when support resembles the HRS lower tail, so these intervals should not
be interpreted as definitive evidence for or against a null effect.

The estimates do not identify a single empirically dominant method.
Cross-fitted RF-AIPW gave 0.141 (95\% CI -0.316 to 0.599) for
actor/direct, 0.309 (-0.280 to 0.898) for spillover, -0.069 (-0.481 to
0.342) for the joint contrast, and -0.520 (-1.436 to 0.397) for
interaction. APIM-style adjusted regression gave 0.192 (bootstrap 95\%
CI -0.072 to 0.451), -0.039 (-0.225 to 0.167), 0.336 (-0.066 to 0.778),
and 0.183 (-0.341 to 0.716), respectively. Every reported interval
included zero at the original 0.005 floor. That fact is descriptive
rather than reassuring: in the nonlinear/near-positivity simulation,
RF-AIPW coverage averaged only 0.773, so the HRS intervals may
themselves be anticonservative.

Point estimates varied materially across estimators. Actor estimates
ranged from -0.019 to 0.241, spillover from -0.268 to 0.309, joint from
-0.325 to 0.649, and interaction from -0.520 to 0.936 at the 0.005
floor. Such disagreement is itself a design diagnostic: the rare 11
state, treatment-model factorization, nonlinear nuisance structure, and
stabilization choices all matter on this dataset. The empirical
application therefore demonstrates how the workflow reveals fragility;
it does not establish that one estimator has recovered the true
retirement effect.

\begin{figure}[H]
\centering
\includegraphics[width=0.88\textwidth]{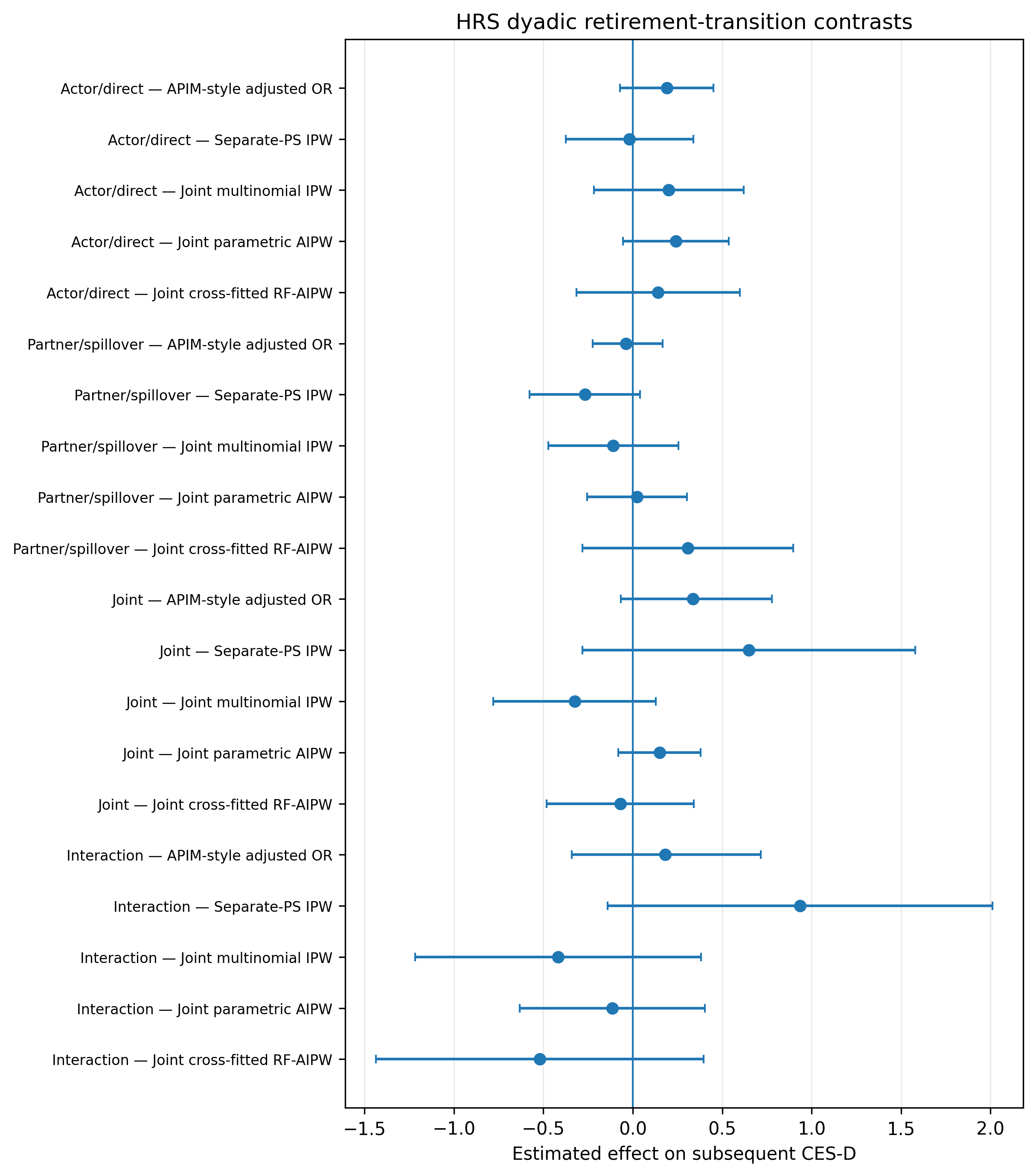}
\caption{HRS estimator estimates and 95\% intervals for actor, spillover, joint, and interaction contrasts. APIM-style intervals are from 2,000 dyad bootstrap resamples; IPW/AIPW intervals use dyad-level pseudo-outcome Wald standard errors. The figure is diagnostic rather than a ranking of methods; sparse support and simulation undercoverage limit substantive interpretation.}
\end{figure}

\subsection{Cohort flow, overlap, and partner-stability
diagnostics}\label{cohort-flow-overlap-and-partner-stability-diagnostics}

The baseline risk sets varied across HRS waves. Before the
first-eligible restriction, eligible couple-trials ranged from 377 in
2018 to 656 in 2010; retaining only the first eligible baseline produced
2,400 independent dyadic units. Administrative member ordering was not
demographically neutral: member 1 was 61.6\% male with mean age 56.1,
whereas member 2 was 38.1\% male with mean age 54.4. The symmetric
reported estimands remain invariant to swapping the two labels because
both orientations are averaged, but member-specific components should
not be interpreted as exchangeable or gender-neutral. Mean baseline
CES-D was 0.94 and 1.04 and mean t+2 CES-D was 0.94 and 1.01 for members
1 and 2, respectively.

Post-baseline partner stability was high but not used to define the
primary cohort. Among the 2,400 baseline-defined couples, 98.1\%
retained the same baseline partner at t+1 and 96.5\% at t+2; 2,311
couples retained the same partner at both follow-ups. This diagnostic
supports the practical coherence of treating the baseline pair as the
dyadic unit while avoiding an eligibility rule that conditions on future
relationship status. A stable-partner-through-t+2 sensitivity
restriction retained 2,313 couples and produced qualitatively similar
conclusions; Appendix E reports the estimates side by side.

Overlap was adequate for the dominant 00 state but substantially weaker
for the rare retirement states. The multinomial model\textquotesingle s
median observed-state propensity was 0.926, reflecting the fact that
most eligible couples remained in 00, while the first percentile was
0.019 and the minimum was approximately 0.005. About 3.3\% of
observations had observed-state propensity below 0.05, and the largest
inverse weight was approximately 201. The cross-fitted random-forest
joint model produced a similar pattern, with 3.7\% below 0.05 and a
maximum weight near 203. Agreement across nuisance approaches suggests
that the weak support is a feature of the data rather than merely an
artifact of the parametric propensity specification.

\subsection{Probability-floor sensitivity and estimator
disagreement}\label{probability-floor-sensitivity-and-estimator-disagreement}

The 0.005 probability floor is binding, so we examined how the
parametric joint estimators changed when the fitted probabilities were
left untruncated or floored at 0.005, 0.010, and 0.020. Actor and
spillover estimates were comparatively stable, whereas the joint and
interaction contrasts moved materially because they depend most heavily
on the rare 11 state. For parametric AIPW, the joint estimate changed
from 0.026 untruncated to 0.150, 0.192, and 0.248 as the floor
increased; the interaction changed from -0.238 to -0.114, -0.087, and
0.012. Joint IPW showed analogous sensitivity. These results make clear
that the floor is not a cosmetic numerical setting in this application.

The HRS pattern should not be read as contradicting the simulation floor
diagnostic in Appendix D.4. A direct code audit confirmed that joint IPW
and parametric AIPW use the same untruncated multinomial probabilities
and, at every threshold c, the same transformation
$g_a^{(c)}=\max\{g_a,c\}/\sum_b\max\{g_b,c\}$. The different trajectories
in Table 8 therefore cannot be attributed to different renormalization.
They arise because the two estimating equations react differently to the
same change in g. IPW uses observed-outcome terms $Y/g_A$, whereas AIPW
modifies residual corrections $(Y-m_A)/g_A$ around an outcome-regression
baseline. For the lowest-propensity observed 11 case driving the first
0-to-0.005 change, both observed CES-D outcomes were zero while the
fitted 11-state means were nonzero. Its raw IPW contribution was
therefore close to zero, but its AIPW residual correction was large;
flooring attenuated that correction and moved AIPW much more than IPW.
Under the focused linear simulation the nuisance models are correctly
specified, so the same operation mainly regularizes finite-sample
residual corrections. The empirical and simulation floor analyses
therefore answer different questions.

Table 8. HRS sensitivity of rare-state contrasts to propensity flooring

\begin{longtable}[]{@{}
  >{\raggedright\arraybackslash}p{(\columnwidth - 6\tabcolsep) * \real{0.2500}}
  >{\raggedright\arraybackslash}p{(\columnwidth - 6\tabcolsep) * \real{0.2500}}
  >{\raggedright\arraybackslash}p{(\columnwidth - 6\tabcolsep) * \real{0.2500}}
  >{\raggedright\arraybackslash}p{(\columnwidth - 6\tabcolsep) * \real{0.2500}}@{}}
\toprule\noalign{}
\begin{minipage}[b]{\linewidth}\raggedright
\textbf{Probability floor}
\end{minipage} & \begin{minipage}[b]{\linewidth}\raggedright
\textbf{Estimator}
\end{minipage} & \begin{minipage}[b]{\linewidth}\raggedright
\textbf{Joint}
\end{minipage} & \begin{minipage}[b]{\linewidth}\raggedright
\textbf{Interaction}
\end{minipage} \\
\midrule\noalign{}
\endhead
\bottomrule\noalign{}
\endlastfoot
Untruncated & Parametric AIPW & 0.026 & -0.238 \\
0.005 & Parametric AIPW & 0.150 & -0.114 \\
0.010 & Parametric AIPW & 0.192 & -0.087 \\
0.020 & Parametric AIPW & 0.248 & 0.012 \\
Untruncated & Joint IPW & -0.324 & -0.418 \\
0.005 & Joint IPW & -0.326 & -0.419 \\
0.010 & Joint IPW & -0.414 & -0.502 \\
0.020 & Joint IPW & -0.501 & -0.327 \\
\end{longtable}

Note. Sensitivity uses one common untruncated multinomial
joint-propensity fit. At each threshold all four fitted probabilities
are floored and then renormalized identically before either IPW or AIPW
is computed. Thus differences between estimator trajectories are not due
to different truncation rules. The cross-fitted RF production run stored
only the floored probabilities, so this table is a parametric diagnostic
of stabilization rather than a redefinition of the target estimand.

The sharp AIPW movement at the first floor is therefore interpretable
rather than anomalous: clipping changes a rare-cell residual correction
even when the corresponding raw observed-outcome contribution is small.
This is precisely why truncation sensitivity should be reported
separately for weighting and augmented estimators.

Estimator disagreement can arise from several sources simultaneously:
sparse support for 11, factorization of the treatment mechanism,
outcome-model restrictions, flexible-learning error, and propensity
stabilization. The simulations help diagnose these possibilities but
cannot identify which empirical estimator is correct. Accordingly, no
method is designated \textquotesingle preferred\textquotesingle{} in the
HRS section. Cross-fitted RF-AIPW is useful because it relaxes
parametric nuisance restrictions; parametric AIPW is useful because its
structure is transparent; joint IPW isolates the treatment-model
component; separate-PS IPW exposes the conditional-independence
assumption; and APIM-style regression shows what a correctly specified
outcome-model strategy would target under the identifying assumptions.

The stable-partner sensitivity provides another form of triangulation.
Restricting to couples observed with the same partner through t+2
changed the RF-AIPW actor estimate from 0.141 to 0.202, the spillover
estimate from 0.309 to 0.388, the joint estimate from -0.069 to -0.154,
and the interaction estimate from -0.520 to -0.743. The intervals
remained wide and crossed zero. This similarity reduces concern that the
primary results are driven by the small fraction of couples whose
partner link changed, although it does not address informative
attrition, death, or unmeasured confounding.

\section{Discussion}\label{discussion}

Dyadic studies often begin from a scientifically sensible intuition: one
person\textquotesingle s exposure may matter for both people. APIM
provides a coherent statistical representation of that interdependence,
but the causal step is separate. Proposition 1 states the restrictive
case in which the familiar treatment coefficients coincide with
standardized causal contrasts. Proposition 2 then shows exactly what is
lost when a dependent four-state treatment law is replaced by the
product of its two marginals, and Corollary 2.1 shows why correct
outcome modeling can still protect an augmented estimator. The
target-trial audit identifies the separate design assumptions needed
before any of these estimators can be given a causal interpretation.

The simulations follow directly from those identities. In the linear
DGP, APIM-style outcome regression was unbiased and had the smallest
RMSE among implementable estimators. Fitted separate-PS IPW deteriorated
as κ increased, and the true-marginal factorization calculation reproduced
essentially the same large-dependence bias after removing marginal
propensity estimation error. This isolates factorization as the dominant
mechanism in that controlled setting. The focused double-robustness
check then confirmed the complementary implication of Corollary 2.1:
parametric AIPW remained near unbiased when either its joint propensity
or outcome nuisance model was correct. Joint treatment modeling is
therefore protection against a specific treatment-mechanism
misspecification, not evidence that a correctly specified APIM outcome
model is intrinsically invalid.

The nonlinear scenarios provide a complementary lesson. Directly
modeling the four-state propensity did not rescue misspecified
parametric nuisance models; APIM-style OR, joint multinomial IPW, and
parametric AIPW were all biased when the confounding functions were
nonlinear. Cross-fitted random forests reduced that bias, but weak
positivity still degraded coverage. This separates three issues that
should not be conflated: the joint versus factorized treatment law,
functional-form adequacy of the nuisance models, and the amount of
empirical support for the requested regime contrast.

The longitudinal simulation, reported in Appendix A, addresses a
different scientific target. Once a time-varying covariate is affected
by prior treatment and also predicts later treatment and outcome,
conventional adjustment can miss a sustained-regime total effect. That
extension motivates caution in genuinely repeated-treatment APIM
applications, but it is not part of the identification argument for the
single retirement transition analyzed in HRS.

\subsection{Relation to the interference
literature}\label{relation-to-the-interference-literature}

The framework sits inside, rather than alongside, the established
interference literature. Hudgens and Halloran (2008) formalized direct
and indirect effects under interference; Tchetgen Tchetgen and
VanderWeele (2012) developed identification tools; Liu et al. (2019)
gave doubly robust estimators under partial interference; Mussavi Rizi
et al. (2024) studied dynamic treatment regimes on dyadic networks; and
Qu et al. (2026) developed semiparametric estimators under heterogeneous
partial interference. Complementary behavioral-science work now
distinguishes multiple ways an intervention itself can be dyadic,
including who participates, what is done, the mechanisms targeted, and
the outcomes considered (Berli et al., 2026). The present paper does not
compete with those contributions. Its role is to translate actor-partner
scientific questions into explicit counterfactual contrasts, formalize
the treatment-factorization assumption in the binary dyadic case, make
emulation gaps auditable, and connect those pieces to an implementable
analysis workflow.

\subsection{Interpretation of the HRS
example}\label{interpretation-of-the-hrs-example}

The HRS application should not be read as a substantive conclusion that
retirement has no effect on depressive symptoms. The both-retire state
contained only 40 couples, the untruncated maximum inverse weight
exceeded 500, estimates depended on propensity flooring, and the
simulation condition most similar to the HRS lower tail produced RF-AIPW
undercoverage. The empirical result is therefore diagnostic: retirement
assignments remained conditionally dependent after extensive baseline
adjustment ($\kappa\approx1.69$), separate-propensity factorization was empirically
implausible, and the 11-related joint and interaction targets were only
weakly supported.

\subsection{Limitations}\label{limitations}

First, the proposed workflow does not solve unmeasured confounding.
Joint conditional exchangeability remains an untestable identifying
assumption, and couples may share retirement planning, financial shocks,
caregiving demands, or health information not fully represented in the
measured baseline history. In particular, health or employment shocks
arising between baseline and the t+1 retirement classification could
affect both retirement and later depressive symptoms; baseline-only
adjustment cannot remove confounding created by such exposure-window
events.

Second, cohort construction conditions on post-baseline information. The
binary retirement contrast excludes dyads when either member is
unemployed, disabled, partially retired, or in another labor-force state
at t+1, and the analysis requires observed t+2 CES-D. These restrictions
define a clear working-versus-retired comparison but can induce
selection if the excluded states or outcome observation are related to
determinants of later depression. We did not use
inverse-probability-of-selection or censoring weights, and
survey-generalization weights were also not applied; the empirical
target is therefore the selected analytic cohort rather than
automatically the full U.S. older-adult population.

Third, positivity is a central constraint because four joint states must
be supported rather than two marginal treatment levels. The HRS
both-retire state contained only 40 couples, and probability-floor
sensitivity materially changed the joint and interaction estimates.
Stabilization can improve numerical behavior, but it changes the
estimating equation and can alter finite-sample bias. Alternative
stochastic targets may be better supported, but they answer a different
scientific question.

Fourth, the outcome is measured at t+2, roughly one biennial wave after
retirement classification at t+1. Retirement effects on mental health
may evolve over time, so the analysis estimates a contrast for this
specific follow-up window rather than a stable or long-run retirement
effect.

Fifth, the cross-fitted random-forest implementation is one practical
nuisance-learning strategy rather than a uniquely optimal algorithm. The
main simulation used fewer folds and trees than the HRS analysis for
computational reasons, and no single learner is guaranteed to dominate
in the rare 11 cell.

Sixth, the simulations were intentionally stylized. They examine
selected failure modes and do not span all forms of assortative mating,
latent shared frailty, measurement error, treatment versions, nonbinary
exposures, informative dyad dissolution, or higher-order networks. The
focused double-robustness check verifies a standard AIPW property but is
not a comprehensive estimator benchmark.

\subsection{A practical workflow for applied dyadic causal
analyses}\label{a-practical-workflow-for-applied-dyadic-causal-analyses}

The framework suggests a concrete workflow that can be used before
fitting any APIM or causal estimator. First, write the exposure as a
joint dyadic intervention and list the treatment states required by the
scientific contrast. If the treatment cannot be meaningfully manipulated
or described as a strategy, the causal language should be narrowed.
Second, define the target population and a dyadic time zero using only
pre-exposure information. Third, draw a temporal causal diagram that
includes both members\textquotesingle{} prior treatments, outcomes, and
shared causes. This step is especially important when one
partner\textquotesingle s prior outcome affects the
other\textquotesingle s later treatment.

Fourth, decide whether the dyad is distinguishable. If roles are
meaningful, report member-specific causal contrasts; otherwise use a
symmetric standardization rather than arbitrarily interpreting one
orientation. Fifth, examine the empirical distribution of the four joint
states before committing to a complex estimator. A rare 11 state can
make the joint and interaction estimands weakly supported even when
actor and spillover contrasts are reasonably estimable. Sixth, choose a
treatment model that represents the joint assignment mechanism. Separate
propensities should be used only when conditional independence is
substantively defensible or when they are deliberately included as a
diagnostic comparator.

Seventh, select nuisance estimators according to the expected complexity
of the data rather than treating machine learning as automatically
superior. Parametric models are efficient and interpretable when well
specified. Cross-fitted flexible learners are attractive when nonlinear
confounding or interactions are plausible and the sample size supports
them. Eighth, diagnose positivity using state counts and estimated
probabilities, not only marginal treatment prevalence. Ninth, if
treatment repeats over time, ask whether any planned adjustment variable
is affected by prior treatment. If so, a longitudinal g-method or
another appropriately defined longitudinal estimand may be required.

Finally, separate three layers when reporting results: the scientific
estimand, the identification assumptions, and the estimator. Statements
such as \textquotesingle the partner effect was 0.31\textquotesingle{}
are incomplete because the same label can refer to a regression
coefficient, a marginal spillover contrast, or a model-specific
conditional contrast. A transparent report names the joint states being
contrasted, the target population over which the contrast is averaged,
the nuisance models used for identification, and the diagnostics that
determine whether the contrast is empirically supported. The TARGET
reporting principles are useful here even when the hypothetical trial is
dyadic rather than individually randomized, because they force the
design and analysis decisions to be stated separately (Cashin et al.,
2025).

\subsection{Extensions and research
directions}\label{extensions-and-research-directions}

Binary exposures are convenient but not essential. When deterministic
joint regimes are weakly supported, stochastic interventions that remain
closer to the observed treatment mechanism may provide alternative
causal targets. Developing such policies for dependent dyadic treatments
is a separate methodological problem and is not pursued here.

The same ideas extend beyond couples. Parent-child dyads,
patient-caregiver pairs, teacher-student pairs, and matched peers can
all be represented as small interference clusters when the scientific
exposure is member-specific and cross-unit effects are plausible. For
households with more than two members, however, the joint treatment
space expands exponentially. Partial-interference estimators that
summarize exposure mappings or use cluster-level policies may be more
appropriate than enumerating every treatment vector. The dyadic case is
valuable precisely because it is the smallest setting in which
interference, joint propensity modeling, and actor-partner language meet
without the combinatorial complexity of large networks.

A further extension is sensitivity analysis for unmeasured shared
causes. The present paper assumes exchangeability after conditioning on
measured dyadic history. Yet shared financial expectations, relationship
quality, caregiving negotiations, or latent health information can
affect both partners\textquotesingle{} treatment decisions and outcomes.
Developing sensitivity analyses parameterized at the dyad level would
complement the joint-treatment framework. Such analyses should
distinguish an omitted common cause that shifts both
members\textquotesingle{} treatment probabilities from an omitted
member-specific cause that primarily affects one treatment-outcome path.

Longitudinal mediation is another open direction. In many substantive
studies, relationship quality, caregiving, or shared behavior is both a
mechanism through which one partner affects the other and a time-varying
confounder for later treatment. APIMeM and related dyadic mediation
models provide rich associational decompositions, while causal mediation
under interference requires cross-world or interventional definitions
that are substantially more demanding. A useful next step is therefore
not to relabel conventional dyadic mediation paths as causal, but to
define longitudinal interventional analogues that remain identifiable
under plausible dyadic histories.

\section{Conclusion}\label{conclusion}

For applied work, the recommendation is constructive. Start from the
scientific actor-partner question and write the corresponding four-state
intervention before fitting a model. Proposition 1 shows when a familiar
APIM coefficient can legitimately carry that causal interpretation;
outside that special case, standardize regime-specific outcome
predictions to the target population. If weighting or augmentation is
used, treat the four-state joint treatment law as the default object.
Proposition 2 shows that multiplying even perfectly known
member-specific propensities discards one conditional dependence
coordinate, δ(H), and can bias causal contrasts in state- and
outcome-specific ways. Diagnose that dependence, joint-state support,
truncation sensitivity, and estimator disagreement as part of the
substantive result. When a deterministic joint state is too sparse,
narrow the causal claim or adopt a better-supported policy estimand. The
goal is not to replace APIM, but to align the scientific intervention,
identification assumptions, treatment mechanism, and estimator.

\section*{Data and Reproducibility}\addcontentsline{toc}{section}{Data and Reproducibility}

The Monte Carlo studies were implemented with fixed seeds, scenario-level checkpoints, raw replication output, summary files, and session information. The true-marginal factorization diagnostic in Table~3A evaluates the Horvitz--Thompson population functional under the exact DGP marginals using 2,000,000 independent standard-normal covariate draws. The fitted-logit population diagnostic solves the population score equations for the two Bernoulli main-effects working models on a separate large Monte Carlo sample and evaluates the resulting factorized functional on 5,000,000 additional draws after applying the simulation's 0.005 probability floor and renormalization; batch Monte Carlo standard errors for the displayed mean absolute biases are below 0.001. The reproducibility materials also include the focused double-robustness diagnostic, realized-propensity calibration, simulation probability-floor sensitivity, the HRS conditional-dependence model, orientation/calendar/age-gap $\kappa$ sensitivity checks, dyad-bootstrap APIM intervals, member-label diagnostics, and an explicit common-normalization audit of HRS propensity-floor sensitivity. The HRS analysis used the public RAND HRS Longitudinal File 2022 (Version 1); redistribution of the source data is governed by HRS/RAND terms, while analysis code and derived nonidentifying summaries can be shared with the manuscript.

\clearpage
\section*{References}
\addcontentsline{toc}{section}{References}

Bang, H., \& Robins, J. M. (2005). Doubly robust estimation in missing
data and causal inference models. Biometrics, 61(4), 962-973.
\url{https://doi.org/10.1111/j.1541-0420.2005.00377.x}

Berli, C., Scholz, U., Allen, J. M., Di Maio, S., Höhener, P. S., Knoll,
N., Luszczynska, A., Nakamura, M. S., Simpson, J. A., Stadler, G.,
Villinger, K., Wilhelm, L. O., \& Rothman, A. J. (2026). What makes an
intervention dyadic? Introducing the DYADIC meta-framework to Describe
Your focAl Dyadic Intervention Components. Annals of Behavioral
Medicine, 60(1), kaaf102. \url{https://doi.org/10.1093/abm/kaaf102}

Cashin, A. G., Hansford, H. J., Hernan, M. A., et al. (2025).
Transparent reporting of observational studies emulating a target trial:
The TARGET Statement. JAMA, 334(12), 1084-1093.
\url{https://doi.org/10.1001/jama.2025.13350}

Chernozhukov, V., Chetverikov, D., Demirer, M., Duflo, E., Hansen, C.,
Newey, W., \& Robins, J. (2018). Double/debiased machine learning for
treatment and structural parameters. The Econometrics Journal, 21(1),
C1-C68. \url{https://doi.org/10.1111/ectj.12097}

Health and Retirement Study. (2025). RAND HRS Longitudinal File 2022
(V1) public use dataset. University of Michigan, Ann Arbor, MI.
Sponsored by the National Institute on Aging (U01AG009740 and
R01AG073289).

Hernan, M. A., \& Robins, J. M. (2016). Using big data to emulate a
target trial when a randomized trial is not available. American Journal
of Epidemiology, 183(8), 758-764. \url{https://doi.org/10.1093/aje/kwv254}

Hernan, M. A., Dahabreh, I. J., Dickerman, B. A., \& Swanson, S. A.
(2025). The target trial framework for causal inference from
observational data: Why and when is it helpful? Annals of Internal
Medicine, 178(3), 402-407. \url{https://doi.org/10.7326/ANNALS-24-01871}

Hernan, M. A., Wang, W., \& Leaf, D. E. (2022). Target trial emulation:
A framework for causal inference from observational data. JAMA, 328(24),
2446-2447. \url{https://doi.org/10.1001/jama.2022.21383}

Hudgens, M. G., \& Halloran, M. E. (2008). Toward causal inference with
interference. Journal of the American Statistical Association, 103(482),
832-842. \url{https://doi.org/10.1198/016214508000000292}

Kenny, D. A., Kashy, D. A., \& Cook, W. L. (2006). Dyadic Data Analysis.
Guilford Press.

Kim, H., \& Kim, J.-S. (2024). Extending the actor-partner
interdependence model to accommodate multivariate dyadic data using
latent variables. Psychological Methods, 29(5), 890-918.
\url{https://doi.org/10.1037/met0000531}

Ledermann, T., Macho, S., \& Kenny, D. A. (2011). Assessing mediation in
dyadic data using the Actor-Partner Interdependence Model. Structural
Equation Modeling, 18(4), 595-612.
\url{https://doi.org/10.1080/10705511.2011.607099}

Liu, L., Hudgens, M. G., Saul, B. C., Clemens, J. D., Ali, M., \& Emch,
M. E. (2019). Doubly robust estimation in observational studies with
partial interference. Stat, 8(1), e214. \url{https://doi.org/10.1002/sta4.214}

Mussavi Rizi, M., Dubin, J. A., \& Wallace, M. P. (2024). Dynamic
treatment regimes on dyadic networks. Statistics in Medicine, 43(30),
5944-5967. \url{https://doi.org/10.1002/sim.10278}

Qu, Z., Xiong, R., Liu, J., \& Imbens, G. W. (2026). Semiparametric
estimation of treatment effects in observational studies with
heterogeneous partial interference. Journal of Business \& Economic
Statistics. Advance online publication.
\url{https://doi.org/10.1080/07350015.2026.2635631}

RAND Center for the Study of Aging. (2025). RAND HRS Longitudinal File
2022 (V1): Documentation. RAND Corporation, Santa Monica, CA.

Robins, J. M., Hernan, M. A., \& Brumback, B. (2000). Marginal
structural models and causal inference in epidemiology. Epidemiology,
11(5), 550-560. \url{https://doi.org/10.1097/00001648-200009000-00011}

Saul, B. C., \& Hudgens, M. G. (2017). A recipe for inference: Start
with causal inference. Add interference. Mix well with R. Journal of
Statistical Software, 82(2), 1-21. \url{https://doi.org/10.18637/jss.v082.i02}

Simon-Tillaux, N., Manzanilla, A., \& Filleron, T. (2026). Target trial
emulation-A unifying approach for causal inference from observational
data. JAMA Network Open, 9(2), e2558192.
\url{https://doi.org/10.1001/jamanetworkopen.2025.58192}

Sonnega, A., Faul, J. D., Ofstedal, M. B., Langa, K. M., Phillips, J. W.
R., \& Weir, D. R. (2014). Cohort profile: The Health and Retirement
Study (HRS). International Journal of Epidemiology, 43(2), 576-585.
\url{https://doi.org/10.1093/ije/dyu067}

Tchetgen Tchetgen, E. J., \& VanderWeele, T. J. (2012). On causal
inference in the presence of interference. Statistical Methods in
Medical Research, 21(1), 55-75. \url{https://doi.org/10.1177/0962280210386779}

VanderWeele, T. J., \& Hernán, M. A. (2013). Causal inference under
multiple versions of treatment. Journal of Causal Inference, 1(1), 1-20.
\url{https://doi.org/10.1515/jci-2012-0002}

Zivich, P. N., \& Breskin, A. (2021). Machine learning for causal
inference: On the use of cross-fit estimators. Epidemiology, 32(3),
393-401. \url{https://doi.org/10.1097/EDE.0000000000001332}

\appendix
\setcounter{figure}{0}
\setcounter{equation}{0}
\counterwithin*{equation}{section}
\renewcommand{\theequation}{\thesection.\arabic{equation}}
\renewcommand{\thefigure}{A\arabic{figure}}
\clearpage
\section{Simulation Data-Generating Mechanisms and Longitudinal Extension}

\subsection{Primary joint-treatment simulation}

For each dyad, U, X1, and X2 were independently standard normal. Let
sc=1.0 under good positivity and 1.85 under near positivity. In the
linear condition, λ1=-0.25+sc(0.55X1-0.20X2+0.45U) and λ2 was obtained
by switching member labels. In the nonlinear condition, λ1 additionally
included sc\{0.55(X1²-1)-0.35 sin(X2)+0.25 U X1\}, with a symmetric
addition for member 2. Four-state treatment probabilities were
proportional to exp(ηs), where η00=0, η10=λ1, η01=λ2, and η11=λ1+λ2+κ.

Member-1 outcomes followed Y1=0.20+0.50A1+0.30A2+0.25A1A2+f1+ε1, with a
symmetric member-2 outcome. In the linear condition
f1=0.70U+0.60X1-0.50X2; nonlinear conditions added 0.70(X1²-1)-0.50
sin(X2)+0.40UX1. Outcome errors had within-dyad correlation 0.40. This
construction yields direct=0.50, spillover=0.30, interaction=0.25, and
joint=1.05.

\subsection{Longitudinal feedback simulation}

The secondary DGP used two joint treatment decisions. First-period
treatment affected L11 and L21; these intermediate variables then
entered the second-period treatment mechanism and final outcome. The
corrected nonlinear second-period treatment model used individual-level
quadratic terms L11² and L21² and did not contain sample-wide centering
quantities. Sustained regimes used the same joint state at both
decisions. The design crossed three sample sizes (750, 1,500, 3,000),
three κ values, two confounding forms, and two positivity conditions for
36 scenarios with 500 replications each.

\subsection{Longitudinal simulation results}

The secondary study crossed n=750, 1,500, and 3,000; κ=0, 0.8, and 1.6;
linear versus nonlinear mechanisms; and good versus near positivity, for
36 scenarios with 500 replications each. Conventional time-varying
regression conditioned directly on treatment-induced L11 and L21,
sequential joint IPW modeled the time-ordered joint treatment mechanism,
and sequential parametric g-computation standardized the outcome after
modeling the treatment-induced covariates. The g-computation
implementation did not compute analytic standard errors, so this
extension is summarized by bias and RMSE rather than treated as a
complete inference comparison.

Table A1. Longitudinal simulation performance averaged across sample
sizes, dependence levels, and four sustained-regime contrasts

\begin{longtable}[]{@{}
  >{\raggedright\arraybackslash}p{(\columnwidth - 8\tabcolsep) * \real{0.2000}}
  >{\raggedright\arraybackslash}p{(\columnwidth - 8\tabcolsep) * \real{0.2000}}
  >{\raggedright\arraybackslash}p{(\columnwidth - 8\tabcolsep) * \real{0.2000}}
  >{\raggedright\arraybackslash}p{(\columnwidth - 8\tabcolsep) * \real{0.2000}}
  >{\raggedright\arraybackslash}p{(\columnwidth - 8\tabcolsep) * \real{0.2000}}@{}}
\toprule\noalign{}
\begin{minipage}[b]{\linewidth}\raggedright
\textbf{Condition}
\end{minipage} & \begin{minipage}[b]{\linewidth}\raggedright
\textbf{Estimator}
\end{minipage} & \begin{minipage}[b]{\linewidth}\raggedright
\textbf{Mean \textbar bias|}
\end{minipage} & \begin{minipage}[b]{\linewidth}\raggedright
\textbf{Mean RMSE}
\end{minipage} & \begin{minipage}[b]{\linewidth}\raggedright
\textbf{Mean coverage}
\end{minipage} \\
\midrule\noalign{}
\endhead
\bottomrule\noalign{}
\endlastfoot
Linear / good & Conventional time-varying regression & 0.577 & 0.612 &
0.245 \\
Linear / good & Sequential joint IPW & 0.033 & 1.190 & 0.950 \\
Linear / good & Sequential g-computation & 0.003 & 0.119 & -\/- \\
Linear / near & Conventional time-varying regression & 0.577 & 0.613 &
0.248 \\
Linear / near & Sequential joint IPW & 0.111 & 3.370 & 0.899 \\
Linear / near & Sequential g-computation & 0.004 & 0.126 & -\/- \\
Nonlinear / good & Conventional time-varying regression & 0.500 & 0.545
& 0.283 \\
Nonlinear / good & Sequential joint IPW & 0.316 & 1.074 & 0.839 \\
Nonlinear / good & Sequential g-computation & 0.088 & 0.182 & -\/- \\
Nonlinear / near & Conventional time-varying regression & 0.477 & 0.525
& 0.303 \\
Nonlinear / near & Sequential joint IPW & 0.672 & 2.123 & 0.759 \\
Nonlinear / near & Sequential g-computation & 0.112 & 0.207 & -\/- \\
\end{longtable}

Note. Coverage is shown for conventional regression and sequential joint
IPW. Sequential g-computation is included as a bias/RMSE benchmark
because analytic standard errors were not implemented for this secondary
extension.

Under linear confounding and good positivity, conventional time-varying
regression had mean absolute bias 0.577, compared with 0.033 for
sequential joint IPW and 0.003 for g-computation. Sequential IPW was
nearly unbiased under the correctly specified linear treatment models
but had much larger RMSE and deteriorated under near positivity. The
exercise demonstrates target mismatch from direct adjustment for
treatment-induced confounders; because the main HRS application is a
point-treatment analysis, these results are intentionally kept outside
the main empirical argument.

\begin{figure}[H]
\centering
\includegraphics[width=0.88\textwidth]{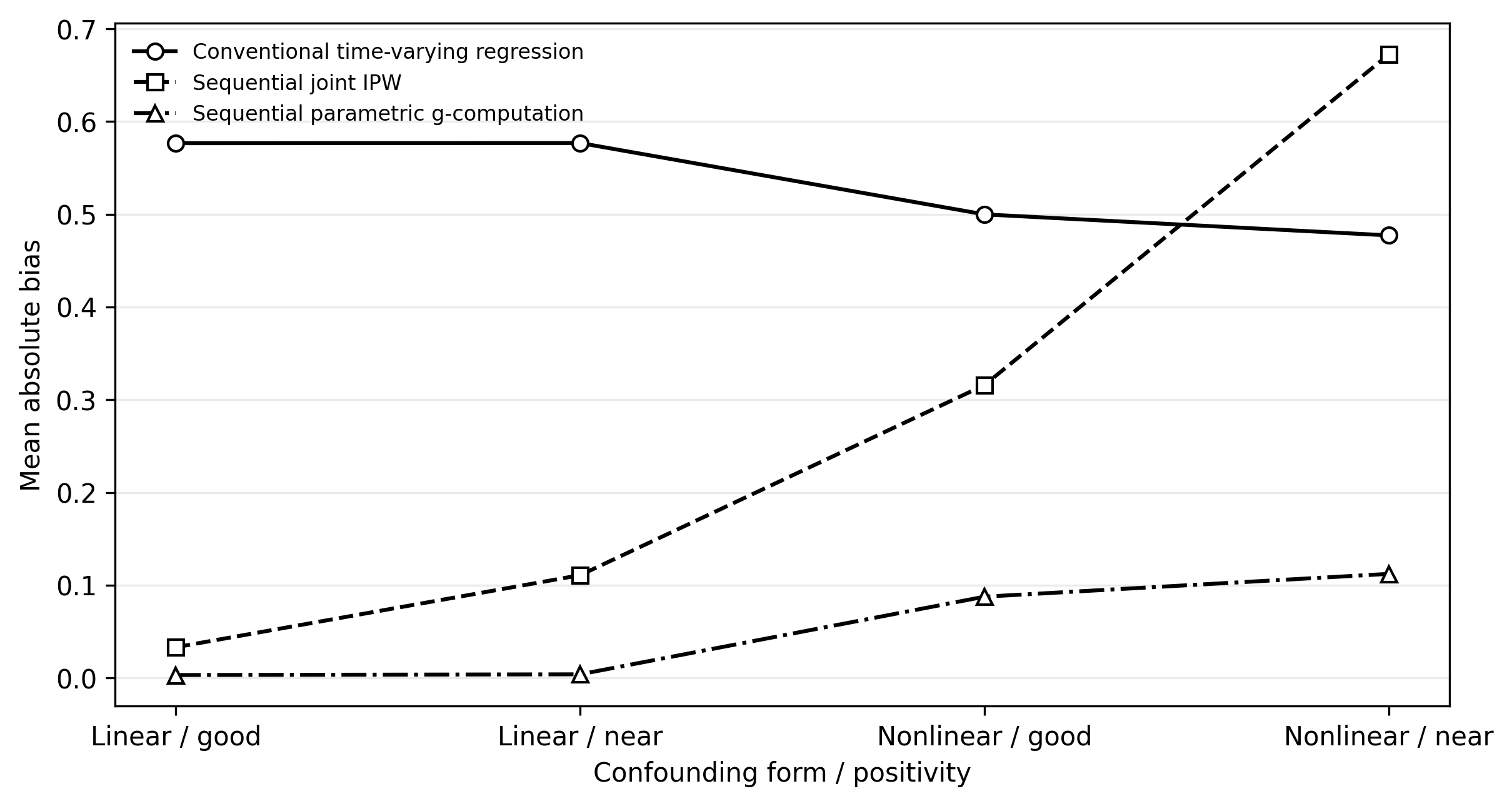}
\caption{Longitudinal simulation mean absolute bias. Conventional regression conditions on treatment-induced covariates and is biased for sustained direct, spillover, and joint effects; sequential methods recover the intended target under their working-model assumptions but can be imprecise under weak positivity.}
\end{figure}

\clearpage
\section{HRS Variable Map and Cohort Flow}

Table B1. RAND HRS variable families used in the empirical illustration

\begin{longtable}[]{@{}
  >{\raggedright\arraybackslash}p{(\columnwidth - 4\tabcolsep) * \real{0.3333}}
  >{\raggedright\arraybackslash}p{(\columnwidth - 4\tabcolsep) * \real{0.3333}}
  >{\raggedright\arraybackslash}p{(\columnwidth - 4\tabcolsep) * \real{0.3333}}@{}}
\toprule\noalign{}
\begin{minipage}[b]{\linewidth}\raggedright
\textbf{Construct}
\end{minipage} & \begin{minipage}[b]{\linewidth}\raggedright
\textbf{RAND variable family}
\end{minipage} & \begin{minipage}[b]{\linewidth}\raggedright
\textbf{Role}
\end{minipage} \\
\midrule\noalign{}
\endhead
\bottomrule\noalign{}
\endlastfoot
Unique respondent ID & HHIDPN & Link respondent records. \\
Wave-specific spouse/partner ID & SwHHIDPN & Construct reciprocal
baseline dyads and post-baseline partner-stability diagnostics. \\
Age & RwAGEY\_E & Baseline eligibility and adjustment. \\
Depressive symptoms & RwCESD & Baseline confounder and t+2 outcome. \\
Chronic conditions & RwCONDE & Baseline adjustment. \\
Labor-force status & RwLBRF & Baseline working eligibility and fully
retired treatment at t+1. \\
Gender & RAGENDER & Baseline adjustment. \\
Education & RAEDYRS & Baseline adjustment. \\
Race & RARACEM & Baseline adjustment. \\
Hispanic ethnicity & RAHISPAN & Baseline adjustment. \\
\end{longtable}

Eligible couple-trial counts by baseline year before retaining the first
eligible baseline per couple were: 609 (1998), 539 (2000), 425 (2002),
528 (2004), 509 (2006), 423 (2008), 656 (2010), 641 (2012), 457 (2014),
557 (2016), and 377 (2018). Retaining only the first eligible baseline
produced 2,400 unique couples.

\clearpage
\section{Estimation Details}

\subsection{Symmetric contrast construction}

For member 1 let $\mu_1(00)$, $\mu_1(10)$, $\mu_1(01)$, and $\mu_1(11)$ denote
the four regime means, and define the analogous quantities for member 2.
Member-1 direct, spillover, joint, and interaction contrasts are
$D_1=\mu_1(10)-\mu_1(00)$, $S_1=\mu_1(01)-\mu_1(00)$, $J_1=\mu_1(11)-\mu_1(00)$,
and $I_1=\mu_1(11)-\mu_1(10)-\mu_1(01)+\mu_1(00)$. For member 2, own
exposure corresponds to the second treatment component, so
$D_2=\mu_2(01)-\mu_2(00)$ and $S_2=\mu_2(10)-\mu_2(00)$. The symmetric actor
and spillover effects are $(D_1+D_2)/2$ and $(S_1+S_2)/2$. Joint and interaction
effects average the member-specific 11-versus-00 and second-difference
contrasts. This orientation-invariant construction is used in the HRS
analysis.

Because the estimands are linear combinations of regime means, the same
coefficient vectors can be applied to estimated influence-function
contributions. If $\phi_j(a)$ denotes the AIPW pseudo-outcome for member $j$
under state $a$, then the empirical influence contribution for the
symmetric actor effect is
$\tfrac12\{\phi_1(10)-\phi_1(00)+\phi_2(01)-\phi_2(00)\}$. The sample mean
yields the point estimate and the sample standard deviation divided by
$\sqrt n$ yields the dyad-level standard error used in the empirical
analysis.

\subsection{Cross-fitting algorithm used in the empirical illustration}

Step 1. Construct one baseline-defined record per eligible couple,
assign the four-level state 00/10/01/11, and retain the two
member-specific outcomes and shared baseline history.

Step 2. Assign complete dyads to K=5 folds using state-stratified random
allocation. No member of a held-out dyad is used to train that
fold\textquotesingle s nuisance models.

Step 3. On the training data, fit a four-category probability forest for
the joint treatment and separate regression forests for member-1 and
member-2 outcomes. The treatment probabilities are bounded below by
0.005 and renormalized to sum to one.

Step 4. For held-out dyads, predict the probability of each joint state.
For each outcome forest, create four counterfactual copies of the
held-out data in which the state is set in turn to 00, 10, 01, and 11;
predict the corresponding conditional outcome means.

Step 5. Construct state-specific AIPW pseudo-outcomes with out-of-fold
nuisance predictions. After all folds are complete, form actor,
spillover, joint, and interaction contrasts at the dyad level and
average them across dyads. Compute Wald intervals from the empirical
variance of the contrast pseudo-outcomes.

\clearpage
\section{Expanded Simulation Diagnostics}

\subsection{Primary simulation: dependence gradient under correct linear specification}
The compact main-text table averages over $\kappa$ and sample size. The dependence gradient is especially informative for the factorization question. Under linear confounding and good positivity, fitted separate-PS IPW mean absolute bias was approximately 0.004 at $\kappa=0$, 0.259 at $\kappa=0.8$, and 0.596 at $\kappa=1.6$. The true-marginal factorization calculation in Table~3A gives 0.000, 0.251, and 0.596 when the true marginal propensities are multiplied, and the fitted-logit population contrast-bias vectors closely track the true-marginal vectors component by component. Thus the large-dependence result is overwhelmingly attributable to false factorization in this DGP rather than to marginal-logit misspecification, and the agreement is not an artifact of averaging offsetting contrast biases. Joint multinomial IPW remained near zero at approximately 0.007, 0.004, and 0.010, while joint parametric AIPW remained approximately 0.006, 0.004, and 0.008.

The APIM-style outcome regression also remained near unbiased in these
scenarios. That behavior is essential to the interpretation of the
study: treatment dependence does not create bias merely because
treatments are correlated. Bias arises when the working models used by
an estimator fail to represent the relevant outcome or treatment
mechanism. A correctly specified outcome regression can remain valid
under conditional treatment dependence, just as a correctly specified
joint propensity estimator can remain valid without modeling the outcome
surface. The causal lesson is therefore about alignment between nuisance
assumptions and the identification strategy, not about replacing APIM
whenever partners\textquotesingle{} exposures are correlated.

\subsection{Nonlinearity, sample size, and machine-learning nuisance estimation}

Under nonlinear confounding and good positivity, cross-fitted RF-AIPW
had mean absolute bias 0.048 averaged across sample sizes, κ values, and
the four estimands, compared with 0.339 for APIM-style outcome
regression, 0.362 for joint multinomial IPW, and 0.377 for parametric
AIPW. RF-AIPW bias decreased as sample size increased: approximately
0.073 at n=500, 0.046 at n=1,000, and 0.024 at n=3,000. This pattern is
consistent with the intended role of flexible nuisance estimation: as
information increases, the forest can approximate nonlinear confounding
functions that are omitted from the parametric working models.

Near positivity changes the picture. Under nonlinear near-positivity
conditions, RF-AIPW retained the lowest average absolute bias among the
methods using estimated nuisance functions but its average coverage fell
well below nominal. The known-propensity IPW benchmark demonstrates why:
even when the true joint treatment probabilities are known, small
probabilities create highly variable weights and large RMSE. Flexible
nuisance learning can reduce model bias but cannot create observations
in unsupported regions of the joint treatment space. This distinction is
important for interpreting real data because improved predictive fit of
a propensity model may coexist with poor overlap.

\subsection{Additional longitudinal diagnostic: target mismatch from conventional adjustment}

The longitudinal DGP makes L11 and L21 consequences of first-period
treatment and causes of later treatment and outcome. As shown in
Appendix A, conventional adjustment therefore targets a different
pathway from the sustained-regime total effect. The interaction happens
to be less biased in the linear DGP because the treatment-induced
pathway is additive in that design; this should not be generalized to
settings with mediated interaction.

This result should not be summarized as \textquotesingle never adjust
for time-varying covariates.\textquotesingle{} The correct statement is
target-specific. If the estimand is a controlled direct effect that
deliberately fixes an intermediate variable, conditioning may be part of
the target definition. If the estimand is the total effect of a
sustained treatment regime, however, directly conditioning on
treatment-induced covariates can block pathways that belong to the
effect while simultaneously trying to control later confounding. The
target trial determines which interpretation is appropriate.

\subsection{Probability-floor sensitivity in the primary simulation}

Because the main simulation also applied a 0.005 floor, we repeated a
focused linear/near-positivity calibration at n=1,000 and κ=1.6 with 500
replications while varying the floor for joint parametric AIPW. Averaged
across the four causal contrasts, the untruncated estimator had mean
absolute bias 0.0079, RMSE 0.202, and coverage 0.951. The corresponding
values were 0.0057, 0.188, and 0.949 at a 0.005 floor; 0.0035, 0.178,
and 0.948 at 0.010; and 0.0012, 0.161, and 0.945 at 0.020. This
controlled result is compatible with the HRS asymmetry in Table 8
because the same floor operation can have different consequences for Y/g
and (Y-m)/g contributions. Section 6.6 documents that both empirical
estimators used the same renormalized propensity matrix.

This calibration is intentionally not used to justify a particular floor
in HRS. Its favorable behavior occurs under a controlled linear
condition with correctly specified parametric nuisance models. The
empirical HRS sensitivity in Table 8 is more consequential because model
error and sparse 11-cell outcome information are both plausible.
Accordingly, the manuscript treats the floor as a tuning sensitivity to
be disclosed, not as a regularization parameter selected from the
simulation.

\clearpage
\section{HRS Diagnostics and Sensitivity Analyses}

\subsection{Baseline cohort and outcome scale}

The primary HRS cohort contains 2,400 couples selected at their first
eligible baseline between 1998 and 2018. The four joint retirement
states were 2,135 couples in 00, 137 in 10, 88 in 01, and 40 in 11. The
empirical contrast therefore relies on a highly unbalanced four-state
treatment distribution. Mean t+2 CES-D scores were approximately 0.94
for member 1 and 1.01 for member 2. RAND constructs RwCESD from eight
binary symptom indicators, so the effect estimates are interpretable as
changes in the expected number of depressive-symptom indicators on that
scale.

The first-percentile observed-state propensity was 0.019 for the
parametric joint model and 0.009 for the cross-fitted random-forest
joint model. The minimum was approximately 0.005 for both approaches
because the analysis used a 0.005 numerical floor. Roughly 3\% to 4\% of
dyads had observed-state propensities below 0.05, and maximum inverse
weights were near 200. These values explain why the 11-related joint and
interaction contrasts are much less stable than the dominant 00 regime
might suggest from marginal summaries alone.

\subsection{Stable-partner sensitivity analysis}

As a descriptive sensitivity analysis, we restricted the cohort to
couples who remained linked to the same partner through t+2. Because
that rule uses future partnership information, it is not used to define
the primary cohort. The restriction produced 2,313 couples;
treatment-state counts were 2,058, 133, 83, and 39 for 00, 10, 01, and
11, respectively. Propensity diagnostics were nearly unchanged.

Table E1. Cross-fitted RF-AIPW estimates in the primary and
stable-partner HRS cohorts

\begin{longtable}[]{@{}
  >{\raggedright\arraybackslash}p{(\columnwidth - 8\tabcolsep) * \real{0.2000}}
  >{\raggedright\arraybackslash}p{(\columnwidth - 8\tabcolsep) * \real{0.2000}}
  >{\raggedright\arraybackslash}p{(\columnwidth - 8\tabcolsep) * \real{0.2000}}
  >{\raggedright\arraybackslash}p{(\columnwidth - 8\tabcolsep) * \real{0.2000}}
  >{\raggedright\arraybackslash}p{(\columnwidth - 8\tabcolsep) * \real{0.2000}}@{}}
\toprule\noalign{}
\begin{minipage}[b]{\linewidth}\raggedright
\textbf{Estimand}
\end{minipage} & \begin{minipage}[b]{\linewidth}\raggedright
\textbf{Primary N=2,400}
\end{minipage} & \begin{minipage}[b]{\linewidth}\raggedright
\textbf{95\% CI}
\end{minipage} & \begin{minipage}[b]{\linewidth}\raggedright
\textbf{Stable-partner N=2,313}
\end{minipage} & \begin{minipage}[b]{\linewidth}\raggedright
\textbf{95\% CI}
\end{minipage} \\
\midrule\noalign{}
\endhead
\bottomrule\noalign{}
\endlastfoot
Actor & 0.141 & {[}-0.316, 0.599{]} & 0.202 & {[}-0.262, 0.667{]} \\
Spillover & 0.309 & {[}-0.280, 0.898{]} & 0.388 & {[}-0.238, 1.013{]} \\
Joint & -0.069 & {[}-0.481, 0.342{]} & -0.154 & {[}-0.576, 0.269{]} \\
Interaction & -0.520 & {[}-1.436, 0.397{]} & -0.743 & {[}-1.687,
0.200{]} \\
\end{longtable}

Note. The stable-partner cohort is a descriptive sensitivity analysis
only. It conditions on retaining the baseline partner through t+2 and is
therefore not used to define the primary target-trial cohort. All
intervals include zero.

The stable-partner estimates are similar in sign and broad magnitude to
the primary results, although the interaction becomes more negative.
Because the restriction removes only 87 couples and leaves the rare 11
cell almost unchanged, it does not resolve the principal overlap
limitation. The sensitivity analysis therefore supports the conclusion
that post-baseline partner turnover is not the main source of empirical
instability, while reinforcing the decision not to define the primary
cohort using future partnership status.

\clearpage
\section{Applied Reporting Checklist for Dyadic Causal Analyses}

The following checklist summarizes the minimum information needed to
distinguish a causal dyadic analysis from an actor-partner regression
that is merely described in causal language. It is intended as an
applied companion to the formal assumptions in the main text and to the
more general TARGET reporting principles.

• Scientific intervention: Define each member-specific exposure and the
joint treatment states or policies being compared.

• Dyadic unit: State how pairs are formed and whether member roles are
distinguishable or symmetrically reported.

• Time zero: Align eligibility, baseline covariate measurement,
treatment assignment, and start of follow-up at the dyad level.

• Potential outcomes: Name the member-specific potential outcomes and
write the actor/direct, partner/spillover, joint, and interaction
contrasts explicitly.

• Interference boundary: Justify why interference is allowed within the
dyad but assumed absent between dyads, or describe the broader cluster
structure if that assumption is not plausible.

• Exchangeability set: Describe measured histories for both members and
shared causes. Avoid selecting covariates solely by statistical
significance.

• Joint positivity: Report counts for all joint states required by the
estimand and diagnostics for estimated observed-state probabilities and
weights.

• Treatment mechanism: Explain whether the joint state is modeled
directly. If separate propensities are multiplied, justify the
conditional-independence assumption.

• Outcome model and standardization: State whether effects are
conditional coefficients or standardized marginal causal contrasts and
identify the target population used for averaging.

• Longitudinal feedback: Determine whether any adjusted variable is
affected by prior treatment. If so, explain the g-method or alternative
estimand used to address it.

• Censoring and dyad dissolution: Define how death, attrition, missing
outcome data, and partner changes enter the target trial and analysis.

• Nuisance estimation: Describe parametric or machine-learning
algorithms, cross-fitting, tuning, and any probability truncation or
bounding.

• Inference: Use dyads as independent analysis units unless a larger
sampling cluster requires additional adjustment. State how standard
errors are obtained.

• Sensitivity and triangulation: Report estimator sensitivity, overlap
limitations, and substantively motivated alternative cohort or modeling
decisions without treating agreement as proof of identification.

• Causal language: Match the strength of the conclusion to the
assumptions and support. A wide interval in a sparse joint state should
not be summarized as evidence of no effect.

\clearpage
\section{Derivations for the Main-Text Identities}\label{sec:proofs}

This appendix gives full derivations of the exact identities used in the main text. The arguments are population identities and do not rely on large-sample approximations. Throughout, expectations are assumed finite, and every working propensity appearing in a denominator is positive on the support relevant to the target.

\subsection{Proof of Proposition~\ref{prop:apim-equivalence}: APIM coefficient equivalence}

\begin{proof}
\textit{Step 1: identification of regime means.}
For a fixed joint state $a=(a_1,a_2)$, define $\mu_1(a)=\E\{Y_1(a)\}$ and $m_1(a,H)=\E(Y_1\mid A=a,H)$. By iterated expectation,
\[
\mu_1(a)=\E_H\!\left[\E\{Y_1(a)\mid H\}\right].
\]
Joint conditional exchangeability gives
\[
\E\{Y_1(a)\mid H\}=\E\{Y_1(a)\mid A=a,H\},
\]
and consistency gives $Y_1(a)=Y_1$ among dyads with $A=a$. Therefore
\begin{equation}
\mu_1(a)=\E_H\!\left[\E\{Y_1\mid A=a,H\}\right]=\E_H\{m_1(a,H)\}.
\label{eq:gformula-proof}
\end{equation}
No functional-form restriction is required for this identification step.

\textit{Step 2: direct and spillover contrasts.}
Under the homogeneous treatment-related conditional mean,
\[
m_1(a_1,a_2,H)=\alpha_1(H)+\beta_Aa_1+\beta_Pa_2+\beta_Ia_1a_2.
\]
The member-1 direct contrast holds the partner exposure at zero. Using \eqref{eq:gformula-proof},
\begin{align*}
D_1
&=\E_H\{m_1(1,0,H)-m_1(0,0,H)\}\\
&=\E_H(\beta_A)=\beta_A.
\end{align*}
Likewise,
\begin{align*}
S_1
&=\E_H\{m_1(0,1,H)-m_1(0,0,H)\}\\
&=\E_H(\beta_P)=\beta_P.
\end{align*}
The baseline surface $\alpha_1(H)$ cancels exactly in both contrasts.

\textit{Step 3: interaction and joint contrasts.}
For the additive interaction,
\begin{align*}
I_1
&=\E_H\{m_1(1,1,H)-m_1(1,0,H)-m_1(0,1,H)+m_1(0,0,H)\}\\
&=\E_H(\beta_I)=\beta_I.
\end{align*}
The $11$-versus-$00$ joint contrast is
\begin{align*}
J_1
&=\E_H\{m_1(1,1,H)-m_1(0,0,H)\}\\
&=\E_H(\beta_A+\beta_P+\beta_I)\\
&=\beta_A+\beta_P+\beta_I.
\end{align*}
This proves the four equalities in Proposition~\ref{prop:apim-equivalence}.

\textit{Step 4: member 2 and symmetric reporting.}
Repeating the argument after switching member labels gives the member-2 analogues. If a label-invariant estimand is defined as the average over both orientations, linearity of expectation implies that the symmetric causal contrast is the average of the two corresponding member-specific contrasts. No additional identification assumption is introduced by this final averaging step.

\textit{Step 5: effect modification.}
If the treatment-related coefficients vary with $H$, write $\beta_A(H)$, $\beta_P(H)$, and $\beta_I(H)$. The same within-$H$ subtraction remains valid, but standardization yields
\[
D_1=\E\{\beta_A(H)\},\qquad
S_1=\E\{\beta_P(H)\},\qquad
I_1=\E\{\beta_I(H)\},
\]
and
\[
J_1=\E\{\beta_A(H)+\beta_P(H)+\beta_I(H)\}.
\]
A single regression coefficient generally represents a projection-weighted conditional association rather than one of these marginal averages. Hence the critical restriction in Proposition~\ref{prop:apim-equivalence} is homogeneity of the treatment-related part of the conditional mean, not simplicity of $\alpha_1(H)$. If the $A_1A_2$ term is omitted, reconstructing the joint effect as actor plus partner further requires $\beta_I(H)=0$ on the target support.
\end{proof}

\subsection{Proof of Proposition~\ref{prop:factorization}: factorization geometry and IPW bias}

\begin{proof}
\textit{Step 1: one-dimensional discrepancy at fixed margins.}
Let
\[
p_1(H)=P(A_1=1\mid H),\qquad p_2(H)=P(A_2=1\mid H),
\]
and define the product-of-marginals law
\[
q_{11}=p_1p_2,\quad q_{10}=p_1(1-p_2),\quad q_{01}=(1-p_1)p_2,\quad q_{00}=(1-p_1)(1-p_2).
\]
Set
\[
\delta(H)=g_{11}(H)-q_{11}(H).
\]
Because $g$ and $q$ have the same two member-specific margins,
\begin{align*}
g_{11}&=q_{11}+\delta,\\
g_{10}&=p_1-g_{11}=q_{10}-\delta,\\
g_{01}&=p_2-g_{11}=q_{01}-\delta,\\
g_{00}&=1-p_1-p_2+g_{11}=q_{00}+\delta.
\end{align*}
Thus, in cell order $(00,10,01,11)$,
\begin{equation}
g-q=\delta(+1,-1,-1,+1).
\label{eq:fixed-margin-direction}
\end{equation}
The fixed-margin discrepancy is therefore one-dimensional for a binary dyad.

\textit{Step 2: covariance, conditional independence, and feasibility.}
Since $A_1A_2=1$ only in cell $11$,
\begin{align*}
\Cov(A_1,A_2\mid H)
&=\E(A_1A_2\mid H)-\E(A_1\mid H)\E(A_2\mid H)\\
&=g_{11}-p_1p_2\\
&=\delta(H).
\end{align*}
Consequently, $\delta(H)=0$ if and only if $g=q$, equivalently $A_1\perp A_2\mid H$. Nonnegativity of the four cells yields the sharp interval
\begin{equation}
-\min\{q_{00},q_{11}\}\le \delta(H)\le \min\{q_{10},q_{01}\}.
\label{eq:delta-bounds}
\end{equation}
Indeed, the lower bounds follow from $g_{00},g_{11}\ge0$ and the upper bounds from $g_{10},g_{01}\ge0$.

\textit{Step 3: total variation.}
Under the convention $\TV(g,q)=\tfrac12\lVert g-q\rVert_1$, conditional total variation is one half of the $\ell_1$ distance between the two four-cell probability vectors. From \eqref{eq:fixed-margin-direction}, every cell differs in absolute value by $|\delta(H)|$. Hence
\begin{align*}
\TV\{g(\cdot\mid H),q(\cdot\mid H)\}
&=\frac12\sum_a|g_a(H)-q_a(H)|\\
&=\frac12\{4|\delta(H)|\}\\
&=2|\delta(H)|.
\end{align*}

\textit{Step 4: regime-specific Horvitz--Thompson bias.}
For a fixed state $a$, assume $q_a(H)>0$ and the displayed expectation is integrable. Let $m_a(H)=\E(Y\mid A=a,H)$. Conditioning on $H$ gives
\begin{align*}
\E\!\left\{\frac{\mathbf 1(A=a)Y}{q_a(H)}\,\middle|\,H\right\}
&=\frac{1}{q_a(H)}\E\{\mathbf 1(A=a)Y\mid H\}\\
&=\frac{g_a(H)m_a(H)}{q_a(H)}.
\end{align*}
Therefore the population Horvitz--Thompson functional based on the factorized denominator is
\[
\widetilde\mu_a=\E\!\left\{\frac{g_a(H)m_a(H)}{q_a(H)}\right\}.
\]
Subtracting $\mu_a=\E\{m_a(H)\}$ gives
\begin{align*}
\widetilde\mu_a-\mu_a
&=\E\!\left\{\frac{\{g_a(H)-q_a(H)\}m_a(H)}{q_a(H)}\right\}\\
&=\E\!\left\{\frac{s_a\delta(H)m_a(H)}{q_a(H)}\right\},
\end{align*}
which is exact. No outcome-model approximation or large-sample argument is involved.

\textit{Step 5: arbitrary linear contrasts and cancellation.}
Let $C=\sum_a c_a\mu_a$ and $\widetilde C=\sum_a c_a\widetilde\mu_a$. By linearity,
\begin{align*}
\widetilde C-C
&=\sum_a c_a\E\!\left\{\frac{s_a\delta(H)m_a(H)}{q_a(H)}\right\}\\
&=\E\!\left[\delta(H)\sum_a\frac{c_as_am_a(H)}{q_a(H)}\right].
\end{align*}
Thus $\delta(H)\neq0$ does not logically force bias for every contrast: the signed, outcome-weighted terms can in principle cancel after averaging over $H$. The proposition identifies the exact condition under which residual dependence is transmitted into a chosen contrast. The displayed identity is for the unnormalized Horvitz--Thompson functional used in the simulation; normalized or stabilized H\'ajek weighting has an additional denominator effect.
\end{proof}

\subsection{Proof of Corollary~\ref{cor:aipw-factorization}: AIPW error decomposition}

\begin{proof}
\textit{Step 1: population AIPW functional.}
Let $\bar g_a(H)>0$ and $\bar m_a(H)$ be deterministic probability limits of working nuisance estimators. Define
\begin{equation}
\Psi_a(\bar g,\bar m)
=\E\!\left[\bar m_a(H)+\frac{\mathbf 1(A=a)\{Y-\bar m_a(H)\}}{\bar g_a(H)}\right].
\label{eq:aipw-pop}
\end{equation}
Conditioning on $H$,
\begin{align*}
\E\!\left[\frac{\mathbf 1(A=a)\{Y-\bar m_a(H)\}}{\bar g_a(H)}\,\middle|\,H\right]
&=\frac{g_a(H)}{\bar g_a(H)}\{m_a(H)-\bar m_a(H)\}.
\end{align*}
Substituting into \eqref{eq:aipw-pop},
\[
\Psi_a(\bar g,\bar m)=\E\!\left[\bar m_a+\frac{g_a}{\bar g_a}(m_a-\bar m_a)\right].
\]

\textit{Step 2: exact error product.}
Subtracting $\mu_a=\E(m_a)$ and collecting terms gives
\begin{align*}
\Psi_a(\bar g,\bar m)-\mu_a
&=\E\!\left[\left(1-\frac{g_a}{\bar g_a}\right)(\bar m_a-m_a)\right]\\
&=\E\!\left[\frac{(\bar g_a-g_a)(\bar m_a-m_a)}{\bar g_a}\right].
\end{align*}
The two population double-robustness cases are immediate: if $\bar m_a=m_a$ almost surely, the error is zero for any positive $\bar g_a$; if $\bar g_a=g_a$ almost surely, the error is zero for any integrable $\bar m_a$.

\textit{Step 3: false-factorization specialization.}
Set $\bar g_a=q_a$. Proposition~\ref{prop:factorization} gives $g_a=q_a+s_a\delta$, so
\[
1-\frac{g_a}{q_a}=-\frac{s_a\delta}{q_a}.
\]
Therefore
\begin{equation}
\Psi_a(q,\bar m)-\mu_a
=-\E\!\left[\frac{s_a\delta(H)\{\bar m_a(H)-m_a(H)\}}{q_a(H)}\right].
\label{eq:aipw-factorization-proof}
\end{equation}
Conditional treatment dependence is therefore not by itself a source of AIPW bias: it matters through its product with outcome-regression error. For $C=\sum_a c_a\mu_a$,
\[
\Psi_C-C=\sum_a c_a\E\!\left[\left(1-\frac{g_a}{\bar g_a}\right)(\bar m_a-m_a)\right].
\]

\textit{Step 4: deterministic product bound.}
Whenever the two factors are square-integrable, Cauchy--Schwarz gives
\begin{equation}
|\Psi_a-\mu_a|
\le \left\|1-\frac{g_a}{\bar g_a}\right\|_2\,\|\bar m_a-m_a\|_2.
\label{eq:product-bound}
\end{equation}
Equation~\eqref{eq:product-bound} makes the product structure explicit, but it is not a semiparametric rate theorem. Cross-fitting, stochastic nuisance estimators, and asymptotic linearity require additional product-rate and regularity conditions. The main paper uses the result only as a population identity explaining the simulation's double-robustness pattern.
\end{proof}

\subsection{The softmax treatment DGP induces nonlinear member-specific marginals}

The fitted Separate-PS simulation arm uses main-effects Bernoulli logistic regressions. The following calculation explains why that arm is not, by itself, a pure factorization experiment when $\kappa\neq0$, and why the true-marginal factorization diagnostic is needed.

\begin{proof}[Derivation]
Under the four-state softmax design,
\[
\eta_{00}=0,\qquad \eta_{10}=\lambda_1(H),\qquad \eta_{01}=\lambda_2(H),\qquad \eta_{11}=\lambda_1(H)+\lambda_2(H)+\kappa.
\]
Let
\[
Z=1+e^{\lambda_1}+e^{\lambda_2}+e^{\lambda_1+\lambda_2+\kappa}.
\]
Then
\begin{align*}
P(A_1=1\mid H)
&=\frac{e^{\lambda_1}+e^{\lambda_1+\lambda_2+\kappa}}{Z}
 =\frac{e^{\lambda_1}\{1+e^{\lambda_2+\kappa}\}}{Z},\\
P(A_1=0\mid H)
&=\frac{1+e^{\lambda_2}}{Z}.
\end{align*}
Taking the log odds cancels $Z$:
\begin{equation}
\logit P(A_1=1\mid H)
=\lambda_1(H)+\log\!\left\{\frac{1+e^{\lambda_2(H)+\kappa}}{1+e^{\lambda_2(H)}}\right\}.
\label{eq:marginal-logit}
\end{equation}
The member-2 expression follows by interchanging $\lambda_1$ and $\lambda_2$. When $\kappa=0$, the logarithmic correction in \eqref{eq:marginal-logit} is identically zero, so the marginal logit for $A_1$ is exactly $\lambda_1(H)$. When $\kappa\neq0$, define
\[
h_\kappa(x)=\log(1+e^{x+\kappa})-\log(1+e^x).
\]
Its derivative is
\[
h_\kappa'(x)=\expit(x+\kappa)-\expit(x),
\]
which is generally nonconstant in $x$. Hence, even when $\lambda_1(H)$ and $\lambda_2(H)$ are linear in covariates, the member-specific marginal logit is generally nonlinear in $H$ for $\kappa\neq0$. A main-effects logistic fit can therefore combine marginal-model misspecification with the separate factorization error. Using the exact true marginals removes the first component while retaining the second, which is precisely the role of the true-marginal factorization diagnostic.
\end{proof}

\end{document}